\documentclass{article}
\usepackage[left=1in,right=1in,top=1in,bottom=1in]{geometry}
\usepackage{amssymb,amsmath,amsthm}
\usepackage[hidelinks]{hyperref}
\usepackage[numbers]{natbib}
\usepackage[title]{appendix}
\usepackage{comment}

\usepackage[affil-it]{authblk}

\newtheorem{thm}{Theorem}
\newtheorem{prop}{Proposition}

\allowdisplaybreaks

\providecommand{\norm}[1]{\lVert#1\rVert}
\providecommand{\abs}[1]{\lvert#1\rvert}
\providecommand{\inner}[2]{\langle#1,#2\rangle}
\providecommand{\C}{\mathbb{C}}
\providecommand{\R}{\mathbb{R}}
\providecommand{\N}{\mathbb{N}}

\newenvironment{psmallmatrix}
  {\left(\begin{smallmatrix}}
  {\end{smallmatrix}\right)}

\title{Quantum Probability via
the Method of Arbitrary Functions}

\author[1]{Liam Bonds}
\author[2]{Brooke Burson}
\author[1,3]{Kade Cicchella}
\author[3]{\\Benjamin H. Feintzeig}
\author[4]{Lynnx}
\author[5]{Alia Yusaini}

\affil[1]{Department of Physics, University of Washington}
\affil[2]{Department of Mathematics, University of Iowa}
\affil[3]{Department of Philosophy, University of Washington}
\affil[4]{Department of Physics, University of Maryland}
\affil[5]{Department of Mathematics, Northeastern University}

\date{}

\begin{document}

\maketitle

\begin{abstract}
    The goal of this paper is to apply the collection of  mathematical tools known as the \emph{method of arbitrary functions} to analyze how probability arises from quantum dynamics.  We argue that in a toy model of quantum measurement the Born rule probabilities can be derived from the unitary Schr\"odinger dynamics when certain dynamical parameters are treated as themselves random variables with initial probability distributions.  Specifically, we study the perturbed double well model, in which the perturbation is treated as a random variable, and we show that for arbitrary initial distributions within a certain class, the dynamics yields the Born rule probabilities in the joint limits given by long times and small values of Planck's constant (the classical limit).  Our results establish the Born rule as a type of universal limiting behavior that is independent of the precise initial dynamical parameters.
\end{abstract}

\section{Introduction}

Quantum mechanics infamously gives rise to probabilistic predictions in measurement settings.  These predictions can be captured by the Born rule, which in one formulation (for projective measurements) asserts that the probability $Pr(P)$ of getting a positive outcome to a measurement of a ``yes-no" proposition represented by a projection operator $P$ is given by the formula
\begin{align}
    Pr(P) = \inner{\Psi}{P\Psi}
\end{align}
when the system is in a state represented by a unit vector $\Psi$ in the Hilbert space corresponding to the system.
While in the orthodox formulation of quantum mechanics, one often sees the Born rule treated as an independent postulate, some have thought that it cries out for an explanation.  One reason for desiring an explanation might be the apparent conceptual gap between the Born rule and the other principles of the theory, which typically are taken to not mention probability whatsoever.  Indeed, proponents of various solutions to the measurement problem have taken it upon themselves to attempt to explain the Born rule. 

Fulfilling this explanatory duty---pressed upon us by the Born rule's success---has been touted as a virtue of certain interpretations of quantum theory.  And so one might wonder: to what extent is it possible to explain the Born rule in the absence of a full interpretation? In other words, is there any relationship between the Born rule and the shared framework of quantum theory that the interpretations have in common?  In this paper, we will argue for a partial positive answer to this question.

We say only a partial positive answer because one might have the view that explanations in general require some recourse to interpretation.  We leave open this possibility in what follows.  What we aim to do specifically is to provide an \emph{explanation schema}, which we hope could be filled in with the details of particular interpretations---of both probability and quantum theory.  We leave these interpretive tasks to future work.  In this paper, we present technical results that we hope could serve as the explanatory core for the approach advocated here.

The explanation schema we pursue here is inspired by methods developed to answer analogous questions in the context of classical physics.  Asking why the Born rule holds of the probabilities assigned to measurement results is just one example of asking why a physical system that is understood as a chance setup produces outcomes with certain probabilities.  Other such physical systems that are standardly understood as chance setups include flipped coins, spun roulette wheels, rolled dice, and more.  For those systems---which are modelled accurately enough using classical physics---one often motivates the need for an explanation by pointing to an apparent and puzzling incompatibility between the probabilistic outcomes and the deterministic physical dynamics.  In quantum theory, this issue does not disappear, but instead is transformed to the puzzle of reconciling the probabilistic nature of measurement outcomes with the (standardly) deterministic dynamics of Schr\"{o}dinger evolution.  Given the similarities between the explanatory question concerning the Born rule at issue here and the analogues in the context of classical chance setups, we believe it is worth investigating whether the methods used to attempt the analogous classical explanations can be imported to the setting of quantum theory.

One collection of resources that have been employed for explaining the values of probabilities assigned to outcomes of classical chance setups is known as the \emph{method of arbitrary functions}.  The rough idea is to take the physical dynamics of the system as fixed, and to show that for arbitrary distributions (in a certain class) over initial parameters governing time evolution, the system will evolve (often in an appropriate limit) to one in which the distribution assigns probability values that are arbitrarily close to the ones we were aiming to explain.  We provide examples to illustrate these methods later on.  Our task is to show that the method of arbitrary functions can be applied to the dynamics of quantum physics as well---specifically, the unitary Schr\"{o}dinger evolution that all no-collapse interpretations of quantum mechanics share---and that doing so leads to a derivation of the Born rule probabilities in certain examples.  While we will only consider particular measurement models, and so we leave much room for further mathematical work in this regard, we hope to at least provide a proof of concept to motivate future work on this approach.

The specific model of quantum measurement that we consider here is the double well potential with the ``flea" perturbation investigated by \citet{LaRe13}.  They provide an interpretation of the model in terms of which the ordinary unitary evolution drives states initially in superpositions arbitrarily close to states with definite measurement outcomes.  The model thus provides a possible mechanism for \emph{effective collapse} of the quantum state under only unitary evolution.  However, \citet{HeWo18} criticize the interpretation of the model as relevant to the measurement problem on several grounds, including for lacking the ability to explain the Born rule.  Our goal is to respond by showing that the method of arbitrary functions allows one  to recover the Born rule, even for superpositions with unequal probabilities.  Thus, our results both bear on a local debate about the status of a particular approach to quantum measurement and provide methods that may be used to explain quantum probabilities more generally.

The structure of the paper is as follows.  In \S\ref{sec:meas} we provide some background concerning explanations of the Born rule in quantum mechanics.  In \S\ref{sec:arbfun}, we review attempts to explain classical probabilities via the method of arbitrary functions and illustrate a first, simplified application of the same mathematical tools to quantum mechanics in the model of a simple harmonic oscillator.  In \S\ref{sec:prob}, we move to our specific model of quantum measurement, applying the method of arbitrary functions to the double well potential with the flea.  Finally, \S\ref{sec:con} concludes with some discussion of the results and some possible future directions.

\section{Measurement and the Born Rule}
\label{sec:meas}

It is now standard in the foundations of quantum mechanics to associate explanations of the Born rule with the measurement problem.  For example, \citet{Ma95} classifies explaining the Born rule as a second problem associated with quantum measurement---the ``problem of statistics"---to be dealt with after one has a theory that explains how measurements yield definite results---the ``problem of outcomes".  However, proponents of different solutions to the measurement problem have sometimes used radically different approaches to explain the Born rule.  (For accessible introductions to different solutions to the measurement problem, see, e.g., \citep{Le16,Ma19,Ba20b}.)


For example, consider the explanation of the Born rule in contemporary Bohmian mechanics.  In Bohmian mechanics---where measurement outcomes are determined by the positions of particles as their motion is guided by the pilot wave---the condition known as \emph{quantum equilibrium} is essential to the explanation of the Born rule.  The condition of quantum equilibrium is that the \emph{initial} distribution of particles (before measurement) matches that given by the squared amplitude of the wavefunction.  If a system is in quantum equilibrium at any given time, then it will remain so for all future times.  The explanation of the Born rule in Bohmian mechanics due to \citet{DuGoZa92} comes in two steps.  First, they attempt to justify the quantum equilibrium condition by appealing to a typicality measure that is equivariant under the dynamics.  Second, they show that quantum equilibrium leads to subsystems like measuring devices typically displaying long run frequencies that approximately match those given by the Born rule.

In contrast, consider the explanation of the Born rule in the contemporary many-worlds interpretation.  The version of the Everettian many-worlds theory developed by \citet{Wa12}---in which worlds are effective quasi-classical structures emerging from the unitary evolution of the wavefunction in decoherence scenarios---employs decision theory to explain the Born rule probabilities.  Wallace relies on the ``principal principle" to treat probability as guiding the beliefs of a rational agent.  He then proves a decision-theoretic representation theorem to show that an agent in a quantum multiverse satisfying certain conditions, which he argues are constraints of rationality, must set their degrees of belief according to the Born rule.

Neither of these Bohmian or Everettian approaches to explaining the Born rule is universally accepted.  For example, \citet{Va91} provides an alternative Bohmian explanation of quantum equilibrium, and \citet{Ru23} criticizes existing Bohmian explanations of the Born rule.  Similarly, alternative approaches to Everettian probability abound \citep{SeCa18}.  In light of the differences between explanations of quantum probabilities within distinct interpretations of quantum mechanics, and given the disagreements among those party to the debates, we believe it is fruitful to pursue the question of whether one could provide an explanation of the Born rule \emph{without committing to an interpretation of quantum mechanics}.  This motivates considering a recent proposal for understanding quantum measurement in the framework of standard quantum theory.

\citet{LaRe13} argue that standard quantum theory has the resources to explain the production of approximately determinate outcomes, which they claim provides a solution to the measurement problem with no modification to quantum theory.  The results are discussed further by \citet{La13,La17}.  Their central idea is that small perturbations  to the potential of a system from the environment can give rise via the standard quantum dynamics to ``effective collapse" of the wavefunction during measurement processes. 
 This proposal has been controversial, and so we will remain agnostic about the claim that this provides a solution to the measurement problem.  Instead, in this paper we will investigate their model without endorsing all of their interpretational claims.  Indeed, since one of the specific criticisms that \citet{HeWo18} level at Landsman and Reuvers is that one cannot provide an adequate explanation of the Born rule in the effective collapse model, we take this as motivation to search for ways of associating probabilities to this system, which is our goal in this paper.

We begin by describing the concrete model employed by \citet{LaRe13}.  The model is a quantum system in a symmetric double well potential with Hamiltonian given by the unbounded operator $H_{0,\hbar}$ on $L^2(\mathbb R)$ defined by
\begin{align}
    H_{0,\hbar} = -\frac{\hbar^2}{2m}\frac{d^2}{dx^2} + V_0(x)\nonumber\\
    V_0(x) = \frac{1}{4}\lambda(x^2-a^2)^2.
\end{align} 
Here, the parameters in the potential $\lambda$ and $a$ are positive constants determining, respectively, the shape of the wells and the width of the barrier between them.  In this case, the ground state wavefunction $\Psi_\hbar^{(0)}$ of the system contains two peaks above $x=+ a$ and $x=-a$.  The ground state wavefunction is approximately Gaussian in a neighborhood of each peak for small values of $\hbar$.  For these reasons, \citet{LaRe13} interpret this ground state as a superposition of the two states localized in each well, analogous to the state in which a measuring device is in a superposition of yielding two outcomes or the state in which Schr\"{o}dinger's cat is in a superposition of being dead and alive.

The problem \citet{LaRe13} consider is how the system might transition from the ground state $\Psi_\hbar^{(0)}$ to a state in which the system is effectively localized in one well, which is analogous to a state in which there is an approximately determinate measurement outcome.  Indeed, one can construct approximately localized states by considering  the first excited state $\Psi_\hbar^{(1)}$ and defining the states
\begin{align}
    \Psi^\pm_\hbar = \frac{\Psi^{(0)}_\hbar \pm \Psi^{(1)}_\hbar}{\sqrt{2}}.
\end{align}
Each of these wavefunctions $\Psi^\pm_\hbar$ has a single peak at $x=\pm a$ and is approximately Gaussian in a neighborhood of that peak for small $\hbar$.  So the task \citet{LaRe13} propose can be reformulated as that of finding a mechanism for the system to evolve from the ground state $\Psi_\hbar^{(0)}$ to one of the localized states $\Psi_\hbar^\pm$.  

To that end, they consider adding a small, localized perturbation---labelled the ``flea"--to the potential of the system. 
\begin{align}
\label{eq:doublewellham}
    H_{\delta,\hbar} = -\frac{\hbar^2}{2m}\frac{d^2}{dx^2} + V_\delta(x), \\
    V_\delta(x) = \frac{1}{4}\lambda(x^2-a^2)^2 +\delta(x)
\end{align}
The flea perturbation $\delta(x)$ is a smooth, compactly supported function of position $x$, whose support does not include the minima of the potential and satisfies certain constraints \citep[][p. 387-388]{LaRe13} Even a small perturbation of this kind can drastically change the shape of the ground state wavefunction.\footnote{Despite calling the flea a ``perturbation", we will not be treating the resulting dynamics using the techniques regularly called ``perturbation theory" \citep[][Ch. 6, 9]{Gr95}. 
 See also \citep{Ka95}.}  They show that for small values of $\hbar$, the ground state of the perturbed Hamiltonian is almost completely localized in a single well as determined by the flea: the state is localized on the opposite side of a positive flea, or on the same side as a negative flea, reflecting the minimization of energy.  Moreover, \citet{LaRe13} provide a dynamical model in which the flea perturbation is added adiabatically to the unperturbed Hamiltonian.  They show that this produces a mechanism in which, for long enough times, the original ground state becomes dynamically localized in one well, thus affecting a transition from a superposition to one of its components.

The sense in which this localization produces an effective collapse, or approximately determinate measurement outcome, is provided through the classical $\hbar\to 0$ limit.  (We will provide details on the mathematical tools for taking the classical $\hbar\to 0$ limit later in \S\ref{sec:arbfun}.)  \citet{LaRe13} pay homage to Bohr's ideas that measurement results should be described using \emph{classical physics}, and so they interpret measurement outcomes as the definite classical states at the positions $x = \pm a$.  The original ground state $\Psi^{(0)}_\hbar$ of the unperturbed Hamiltonian, which is a superposition of different localized states, has as its classical limit the classical mixed state given by a probabilistic mixture, which we denote for points $(x,p)$ in phase space as
\begin{align}
    \rho^{(0)}_0 = \frac{1}{2}\Big(\delta_{(+a,0)} + \delta_{(-a,0)}\Big),
\end{align}
where $\delta_{(\pm a,0)}$ is the delta-function or point-mass distribution centered at the phase space point with position $x = \pm a$ and momentum $p=0$ \citep[][p. 382]{LaRe13}.\footnote{The point mass distributions $\delta_{(\pm a, 0)}$ should not be confused with the flea perturbation function $\delta(x)$ which is denoted without any subscript.}  The classical limit of the unperturbed ground state $\rho^{(0)}_0$ does not correspond to a single phase space point, and indeed it has nonzero variance in position. 
 \citet{LaRe13} interpret $\rho^{(0)}_0$ as a state in which there is \emph{no} classical determinate measurement outcome.  On the other hand, the ground state of the perturbed Hamiltonian becomes localized in one well, approaching one of $\Psi^{\pm}_\hbar$ (as determined by the location of the perturbation).  It's classical limit is the classical pure state given by the probability measure
\begin{align}
    \rho^{\pm}_0 = \delta_{(\pm a, 0)},
\end{align}
at the phase space point $x=\pm a$ and $p=0$, which signifies that the wavefunction is localized around a single peak at $x=\pm a$ 
\citep[][p. 388-9]{LaRe13}. Since the classical limit of the perturbed ground state $\rho^{\pm}_0$ corresponds to a single phase space point, \citet{LaRe13} interpret it as one in which there \emph{is} a classical determinate measurement outcome.

As mentioned before, the work of Landsman and Reuvers has not been widely accepted as a solution to the measurement problem. In particular, \citet{HeWo18} provide a number of challenges to the interpretation and generalization of the results in the double well model, which we can describe now that we have presented the basics of Landsman and Reuver's model.  Do effective or approximate measurement outcomes in the $\hbar\to 0$ limit of the model suffice to explain the appearance of actual determinate measurement outcomes for physically realistic values $\hbar>0$?  Can one justify treating the environmental influence via a perturbation to the potential rather than an interaction with a further quantum system?  Can one achieve similar results in models that allow for more than two possible ``outcome states" in a measurement process?  We set aside these important questions for future work.  In what follows, we remain agnostic about whether the results of \citet{LaRe13} provide a solution to the measurement problem.  Instead, we take the double well potential as only an interesting model to investigate for the tools we develop to explain the Born rule.

We note that our investigation is motivated by the current state of the literature because \citet{HeWo18} also argue that the explanation that \citet{LaRe13} offer for the Born rule is inadequate.  Landsman and Reuvers claim that since the location of the flea perturbation determines the measurement outcome on their interpretation, the statistics of measurement outcomes are determined by the statistics of the flea.  Put another way, our uncertainty about the outcome of the measurement is determined by our uncertainty about the position of the flea within the potential.  Thus, they explain the equally weighted Born rule statistics produced by the unperturbed ground state superposition by arguing from a principle of indifference; since we have no reason to expect the flea will favor one side over the other, we should employ a probability distribution for the flea that gives equal weight to both sides, and hence yields equal probabilities for the two measurement outcomes.  \citet{HeWo18} rightly point out that this explanation of the Born rule probabilities is unsatisfactory because it makes the probabilities depend \emph{only} on the distribution or uncertainty governing the location of the flea perturbation and \emph{not} on the initial quantum state in the superposition.  For example, this explanation leaves open the question: if the system started in a superposition
\begin{align}
    \Psi_\hbar = \alpha\Psi^+_\hbar + \beta\Psi^-_\hbar
\end{align}with \emph{unequal weights} $|\alpha|^2\neq |\beta|^2$ governing the different localized components, why should we assign unequal probabilities for the corresponding measurement outcomes?

We do not take this to invalidate the significance of work in the double well model.  Rather, we take this as a call to action for further investigation of probabilities in this framework.  In the remainder of the paper, we hope to provide an improved explanatory schema for the Born rule probabilities, and one that can yield explanations even for unequal weights in the double well model.  We now turn to the mathematical tools we will use for our derivation of the Born rule.

\section{The Method of Arbitrary Functions}
\label{sec:arbfun}

The schema we propose for explaining the Born rule stems from the \emph{method of arbitrary functions} as it has been used to explain probabilities for outcomes in classical chance setups. 
 The method of arbitrary functions provides tools for showing how physical dynamics lead to universal probability assignments, i.e., probability assignments that are at least approximately the same for a range of (initial) conditions.  The general idea is to show that for \emph{arbitrary} uncertainty about the initial state of a system or parameters in the dynamical equations, the distribution of the system will evolve in some approximate or limiting regime toward the \emph{same} final distribution.
 
 For example, \citet{Po96} analyzes the spinning of a roulette wheel, in which any smooth enough distribution over the initial angular velocity produces the result, in the regime of high enough velocities, of approximately equal probabilities for the wheel to land on each outcome.  Similarly, \citet{DiHoMo07} analyze the physical dynamics of flipping a coin and show that smooth enough distributions over the initial orientation and angular velocity also lead to the result, in the regime of high enough velocities, of approximately equal probabilities for heads and tails.  \citet{Ho34} discusses further examples, including rolling dice.  \citet{En92} provides comprehensive mathematical background that unifies many of these results in the context of classical physical examples.

Philosophers have also shown interest in the method of arbitrary functions for its role in understanding chance \citep{vP83}.  Some---including \citet{Ab10}, \citet{St11}, \citet{Ro11a}, and \citet{Ho16}---have argued that the method of arbitrary functions provides a physical foundation for objective probabilities.  On the other hand, \citet{My21} argues that the method of arbitrary functions only provides constraints from physics on what rational agents should agree to believe about outcomes of chance setups.  And \citet{dC22} criticizes the idea that the probabilities derived from the method of arbitrary functions could be sufficiently objective.  We take such philosophical discussions, although still controversial, to be evidence for the potential of the method of arbitrary functions to bear on foundational issues regarding physical probability.  In what follows, we remain agnostic about the precise interpretation given to the method of arbitrary functions, especially regarding subjective and objective notions of probability.

In the remainder of this section, we illustrate the application of the method of arbitrary functions to a simple example: the harmonic oscillator.  We first review the classical version of the harmonic oscillator, which has been previously studied, and which will allow us to state certain technical results that we will use in the rest of our investigation.  Then we show how to apply the method to the quantum version of the harmonic oscillator.  This will serve as both a proof of concept and an opportunity to introduce the mathematical tools we will need for considering the classical limit in later results.

For the classical harmonic oscillator in one dimension, e.g., a mass on a spring, the equation of motion for a mass starting from rest is given by
\begin{equation}
x(t) = x(0)\cdot \cos(\omega t),
\end{equation}
where $x(t)$ is the position of the mass at time $t$, $x(0)$ is its initial displacement, and $\omega$ is its frequency of oscillation determined by the inertial mass and the so-called spring constant.  If one has uncertainty about any of the initial conditions or dynamical parameters, then one can treat $x(t)$ as a random variable.  We consider specifically the case where we have uncertainty about the frequency and treat $\omega$ as a random variable.  The crucial mathematical ingredient is the realization that since the motion is periodic, we can write
\begin{equation}
\cos(\omega t) = \cos\left( (\omega t) (\text{mod }2\pi)\right).
\end{equation}
This focuses our attention on the random variable $(\omega t) (\text{mod }2\pi)$.  \citet[][Thm. 3.2 and Thm. 3.5]{En92} then provides the following fundamental result.
\begin{thm}
\label{thm:inftime}
    If the random variable $\omega$ possesses a density with respect to the Lebesgue measure on $\R^+$, then in the limit $t\to \infty$, the random variable $(\omega t) (\text{mod }2\pi)$ approaches a uniform distribution on the interval $[0,2\pi]$ in both the variation distance and the weak* topology.\footnote{A sequence of measures $(\mu_n)$ on $\mathbb R^n$ converge to a measure $\mu$ in the weak* topology if the expectation values for all smooth compactly supported functions $f$ converge, i.e. $\int f(x) d \mu_n(x) \to \int f(x) d \mu(x)$ for all $f \in C^\infty_c(\mathbb R^n)$.}
\end{thm}
\noindent \citet[][\S3]{En92} determines explicit bounds on the rate of convergence for a random variable whose density has bounded variation.  One consequence we wish to highlight is that in the $t\to\infty$ limit, the expectation value of $x(t)$ approaches $0$.  This establishes that a large class of uncertainties about the initial parameters of the system all approach the same distribution, and same expectation values, in the limit of long times.

We now provide our first novel illustration of the use of the method of arbitrary functions in quantum mechanics by showing that the same methods can be applied to the quantum harmonic oscillator.  The quantum harmonic oscillator is described by a wave function $\Psi\in L^2(\R)$ satisfying the Schr\"odinger equation for the Hamiltonian
\begin{align}
    H^\hbar_{m,\omega} = -\frac{\hbar^2}{2m}\frac{\partial^2}{\partial x^2} + \frac{1}{2}m\omega^2 x^2,
\end{align}
where $\hbar$ is Planck's constant.  The time evolution of the quantum state yields
\begin{align}
    \Psi^\hbar_{m,\omega}(t) = e^{-iH^\hbar_{m,\omega} t/\hbar}\Psi
\end{align}
for all $t>0$, where the values of Planck's constant $\hbar$, the mass $m$, and the frequency $\omega$ are fixed.

To analyze this system, we will use the classical $\hbar\to 0$ limit.  To that end, we employ the tools given by \emph{Weyl quantization} and \emph{Wigner functions} \citep{La98b}.  The Weyl quantization maps $\mathcal{Q}_\hbar: C_c^\infty(\R^2)\to \mathcal{B}(L^2(\R))$ take classical observables $f\in C_c^\infty(\R^2)$ on the phase space $\R^2$ and transform them to quantum operators according to the definition
\begin{align}
    \left(\mathcal{Q}_\hbar(f)\Psi\right)(x) = \int_{\R^2} \frac{dqdp}{2\pi\hbar} e^{ip\cdot (x-q)/\hbar} f\left(\frac{1}{2}(x+q), p \right) \Psi(q)
\end{align}
for all $\Psi\in L^2(\R)$.
The Wigner function associates any quantum state given by a wavefunction $\Psi$ to a classical quasi-probability density $W_\hbar^\Psi: \R^2\to \R$ on the classical phase space given by
\begin{align}
    W_\hbar^\Psi(x,p) &= \frac{2}{\hbar}\int_\R \overline{\Psi(x+y)}\Psi(x-y) e^{2ip\cdot y/\hbar}dy\\
    &= \frac{1}{\hbar}\inner{\Psi}{\Omega_\hbar^W(x,p)\Psi},
\end{align}
where $\Omega_\hbar^W(x,p)$ is the linear operator
\begin{align}
\label{eq:Omega}
\big(\Omega_\hbar^W(x,p)\Psi\big)(y) = 2e^{2ip(y-x)/\hbar}\Psi(2x-y)
\end{align}
for all $\Psi\in L^2(\R)$.
Putting together Weyl quantization and the Wigner function, the quasi-probability density can be used to mimic classical expectation values for observables that agree with the quantum expectation values of their quantized counterparts by satisfying the equality
\begin{align}
    \int_{\R^2} f(x,p)W_\hbar^\Psi(x,p)  dxdp = \left\langle \Psi,\mathcal{Q}_\hbar(f)\Psi\right\rangle
\end{align}
for all $\Psi\in L^2(\R)$ and $f\in C_c^\infty(\R^2)$.\footnote{We use the convention of anti-linearity on the left, i.e. $\langle a \psi, \varphi \rangle = \overline a \langle \psi, \varphi \rangle $ where $\overline a$ is the complex conjugate of $a \in \C$.}  If there is a probability measure $\mu_0^\Psi$ on $\R^2$ such that
\begin{align}
    \lim_{\hbar\to 0} \int_{\R^2} f(x,p) W_\hbar^\Psi(x,p) dxdp = \int_{\R^2} f(x,p) d\mu_0^\Psi,
\end{align}
for all $f\in C_c^\infty(\R^2)$ (i.e., $W_\hbar^\Psi dxdp$ approaches $\mu_0^\Psi$ in the weak* topology), then we will call $\mu_0^\Psi$ the \emph{classical limit} of the quantum state $\Psi$.

One useful feature of Weyl-Wigner quantization for analyzing the harmonic oscillator is that it commutes with the corresponding time evolution in the sense that
\begin{align}
\label{eq:dyncommutes_0}W_\hbar^{\Psi^\hbar_{m,\omega}(t)}(x,p) = W_\hbar^\Psi(x_{m,\omega}(-t),p_{m,\omega}(-t))
\end{align}
for all $t\in\R$, where the time evolution occurs with fixed values for the mass $m$ and the frequency $\omega$.  Here, $\Psi^\hbar_{m,\omega}(t)$ represents the time-evolution of the initial quantum state $\Psi$ by the Schr\"{o}dinger equation for some time $t>0$, whereas the phase space point $(x_{m,\omega}(-t),p_{m,\omega}(-t))$ represents the reversed time evolution of the initial classical state $(x,p)$ by the classical harmonic oscillator dynamics for the same time $t>0$.  That is, $(x_{m,\omega}(-t),p_{m,\omega}(-t))$ is the classical state that would evolve into $(x,p)$ after a time $t$ under the corresponding classical harmonic oscillator dynamics.  We now explicitly include the subscripts $m$ and $\omega$ on the right hand side to denote that the classical time evolution uses the dynamics with the same fixed value of the mass $m$ and the frequency $\omega$ that appears on the left hand side.

To apply the method of arbitrary functions to the quantum harmonic oscillator, we analyze both the limit of long times and the classical limit under the supposition that the mass $m$ and the frequency $\omega$ are random variables.  We suppose that $m\omega$ is a constant in order to fix the possible classical trajectories.  Moreover, we suppose that the random variable $\omega$ has a density on $\R^+$.  Then, for each point $(x,p)\in\R^2$, the value of the time-evolved Wigner function $W_\hbar^{\Psi^\hbar_{m,\omega}(t)}(x,p)$ is a random variable, where the evolution of $\Psi$ is determined by the values of $m$ and $\omega$.  We have the following result for the limiting behavior of the quantum harmonic oscillator: in the combined classical and long time limits, the Wigner function becomes uniformly distributed over classical orbits in phase space.

\begin{prop}
   Suppose that $\omega$ and $m$ are random variables with a joint distribution, and that $m\omega$ is constant.  Suppose further that the random variable $\omega$ has a density $d\omega$ on $\mathbb{R}^+$.  Let $\Psi\in L^2(\R)$ have a classical limit $\mu_0^\Psi$ with density $W^\Psi_0$ with respect to the Lebesgue measure, so $d\mu_0^\Psi = W^\Psi_0 dxdp$.  Then we have
    \begin{align}
        \lim_{\hbar\to 0}\lim_{t\to\infty}\int_{\R^+}W^{\Psi^\hbar_{m,\omega}(t)}_\hbar(x,p)d\omega = \int_0^{2\pi} W_0^{\Psi}(x_1(t),p_1(t)) dt
    \end{align}
    in the weak* topology.  On the right hand side, $(x_1(t),p_1(t)) := (x_{1,1}(t),p_{1,1}(t))$ denotes the classical time-evolved state with fixed mass $m=1$ and frequency $\omega = 1$ and hence period $2\pi$.
\end{prop}

\begin{proof}
Consider $\Psi\in L^2(\R)$ with a classical limit denoted as in the statement of the proposition.  Note that generically, we have
\begin{align}
    \label{eq:HOdynamics}
    x_{m,\omega}(t) = x \cos(\omega t) + \frac{p}{m\omega} \sin(\omega t) && p_{m,\omega}(t) = -m\omega x \sin(\omega t) + p \cos(\omega t).
\end{align}
So we have for any $f\in C_c^\infty(\R^2)$,
\begin{align}
\lim_{t\to\infty}\int_{\R^2}\int_{\R^+} f(x,p)W_\hbar^{\Psi^\hbar_{m,\omega}(t)}(x,p)  d\omega dx dp &= \lim_{t\to\infty}\int_{\R^2}\int_{\R^+} f(x,p)W_\hbar^{\Psi}(x_{m,\omega}(-t),p_{m,\omega}(-t)) d\omega dx dp
\label{eq:dyncommutes}\\
&=\int_{\R^2}\lim_{t\to\infty}\int_{\R^+}f(x,p)W_\hbar^{\Psi}(x_{m,\omega}(-t),p_{m,\omega}(-t)) d\omega dx dp\label{eq:dominatedconvergence}\\
&=\int_{\R^2}\int_{0}^{2\pi} f(x,p)W_\hbar^\Psi(x_1(t),p_1(t)) dt dx dp\label{eq:HOinftime}\\
&=\int_{0}^{2\pi}\int_{\R^2}f(x,p)W_\hbar^\Psi(x_1(t),p_1(t)) dx dp dt\label{eq:fubini}.
\end{align}
In lines (\ref{eq:HOinftime})-(\ref{eq:fubini}), $\omega=1$ becomes a constant, whose sole role is to determine the period $2\pi$.  Line (\ref{eq:dyncommutes}) is justified by Eq. (\ref{eq:dyncommutes_0}), while line (\ref{eq:dominatedconvergence}) is justified by the dominated convergence theorem since, for fixed $\hbar\in (0,1]$, the integrand is bounded by
\begin{align*}
\left|\int_{\R^+}f(x,p)W_\hbar^{\Psi}(x_{m,\omega}(-t),p_{m,\omega}(-t)) d\omega\right|\leq \frac{2}{\hbar} \cdot \sup_{(x,p)\in \R^2} |f(x,p)|. 
\end{align*}
The crucial step is line (\ref{eq:HOinftime}), which is justified by Thm. \ref{thm:inftime}---the fact that the random variable $(\omega t) (\text{mod }2\pi)$ approaches a uniform distribution on $[0,2\pi]$ implies that expectation values of the classical periodic motion in Eq. (\ref{eq:HOdynamics}) can be determined by averaging over a period. Since the result is an average over the period of trigonometric functions we can replace $-t$ with $t$. Finally, line (\ref{eq:fubini}) is obtained from the Fubini-Tonelli theorem since, for fixed $\hbar \in (0,1]$, the integral is bounded by
\begin{align}
\int_{\R^2}\int_{0}^{2\pi}\abs{f(x,p)W_\hbar^\Psi(x_1(t),p_1(t))}dtdxdp\leq \frac{2 \pi}{\hbar} \int_{\R^2}\abs{f(x,p)} dxdp,
\end{align}
which is finite since $f\in C_c^\infty(\R^2)$.

We also know that $\hbar\mapsto \norm{\mathcal{Q}_\hbar(f)}_\hbar$ is continuous, and hence bounded for $\hbar\in [0,1]$, so we can apply the dominated convergence theorem to find that
\begin{align}
\lim_{\hbar\to 0}\lim_{t\to\infty}\int_{\R^2}\int_{\R^+} f(x,p)W_\hbar^{\Psi^\hbar_{m,\omega}(t)}(x,p)  d\omega dx dp &=
\lim_{\hbar\to 0}\int_{0}^{2\pi}\int_{\R^2} f(x,p)W_\hbar^\Psi(x_1(t),p_1(t)) dxdpdt\\
&= \int_{0}^{2\pi}\int_{\R^2}f(x,p)W_0^\Psi(x_1(t),p_1(t)) dxdpdt
\end{align}
since the integrand is bounded by
\begin{align}
    \Big|\int_{\R^2}f(x,p)W_\hbar^\Psi(x_1(t),p_1(t))dxdp\Big| \leq \sup_{\hbar\in [0,1]}\norm{\mathcal{Q}_\hbar(f)}_\hbar.
\end{align} 
This is what we set out to show.  (In fact, it follows from what we have shown that the $t\to\infty$ and $\hbar\to 0$ limits commute in this case.)
\end{proof}

This result provides a first application of the method of arbitrary functions to quantum systems.  The example illustrates how one can control the limit of long times $t\to\infty$ and the classical limit $\hbar\to 0$ so as to extract universal limiting behavior.  The behavior is universal in that it captures a wide range of initial distributions over the unknown frequency $\omega$, all of which will approach the same limiting distribution in long times.  Physically, we see that the $t\to\infty$ limit yields averages over periodic motions, while the $\hbar\to 0$ limit suppresses quantum interference terms. 
 For example, if $\Psi$ is a coherent state, then the Wigner function $W^\Psi_\hbar$ is a Gaussian whose width is proportional to $\hbar$.  So while the classical limit $W_0^\Psi$ is a delta function at a point $(x,p)$ in phase space, the combined classical and long time limits smear out $W_0^\Psi$ to become a uniform distribution over the classical orbit through $(x,p)$ for the fixed value of $m\omega$.  We emphasize that since the classical and long time limits commute in this case, the result holds even when the time evolution is taken according to the quantum dynamics.  Employing the quantum, rather than classical, equations of motion is our central conceptual innovation beyond the classical method of arbitrary functions.
 
 We have shown that under certain assumptions, the combined $t\to\infty$ and $\hbar\to 0$ limits of the quantum probabilities leads to universal behavior for the quantum harmonic oscillator.  In this case, the universal behavior we point to is the resulting uniform distributions along classical trajectories.  Indeed, we wish to show that for a toy model of quantum measurement, the same limits produce a different kind of universal behavior.  In the models we consider next, the method of arbitrary functions along with the $t\to \infty $ and $\hbar\to 0$ limits yield the Born rule probabilities as universal probability values for the measurement outcomes.  We now turn to the toy model of quantum measurement governed by the double well potential.

\section{Probability from the Flea}
\label{sec:prob}

In this section, we will show how Born rule measurement probabilities can be derived from Schr\"odinger time evolution in the toy model of quantum measurement due to \citet{LaRe13}, given certain auxiliary assumptions. Recall that Landsman and Reuvers model quantum measurement using the one-dimensional symmetric double well potential with the ``flea" perturbation $\delta$:
\begin{align}
    V_\delta(x) = \frac{1}{4}\lambda(x^2-a^2)^2 +\delta(x),
\end{align}
where $\delta$ is a smooth, compactly supported function. The perturbed double well potential is intended to represent the combined measured system and apparatus in a binary measurement which can take values $x = \pm a$. The values of $x$ away from $\pm a$ represent other configurations the system can take in principle, while the measurement outcomes correspond only to the minima of the wells. Landsman and Reuvers show that, given a system whose initial state is approximately a superposition of Gaussians centered on the two wells, taking the $t \to \infty$ (with the perturbed double well dynamics) and $\hbar\to 0$ limits of this state  results in either a state totally localized at $x=+a$ or one totally localized at $x=-a$. With their toy model, Landsman and Reuvers provide a mechanism that captures one aspect of quantum measurement: wavefunction collapse.

In this section, we aim to capture the other notable aspect of quantum measurement: that probabilities of measurement outcomes are determined by the Born rule. We use the same perturbed double well model as Landsman and Reuvers, but we consider a larger class of initial states.  We now present some general features of the system to establish background and notation.  As discussed in Section \ref{sec:meas}, symmetric and anti-symmetric superpositions of the ground state $\Psi^{(0)}_\hbar$ and first excited state $\Psi^{(1)}_\hbar$ result in states which are approximately localized in the right and left wells, respectively:
\begin{align} \label{eqn:loc}
    \Psi^\pm_\hbar = \frac{\Psi^{(0)}_\hbar \pm \Psi^{(1)}_\hbar}{\sqrt{2}}.
\end{align}
Many of the qualitative features of the double well and the perturbed double well models follow from the phenomenon of \emph{asymptotic degeneracy}, which has been studied in detail in double well systems \citep{Ha78,Ha80,HeSj84,JoMaSc81,JoMaSc81a,Si85}: letting $\Delta_\hbar = E_\hbar^{(1)}-E_\hbar^{(0)}$ denote the difference of the energy eigenvalues for the eigenstates $\Psi_\hbar^{(1)}$ and $\Psi_\hbar^{(0)}$ of the unperturbed Hamiltonian, we have: 
\begin{align} \label{eqn:Delta_hbar}
    \Delta_\hbar \sim \hbar \sqrt{\frac{2a^2\lambda}{e \pi}} e^{-d_V/\hbar} \enspace (\hbar \to 0),
\end{align}
where $d_V$ is the typical WKB-factor
\begin{align}
    d_V = \int_{-a}^{a} dx \sqrt{V(x)}.
\end{align}
We will consider initial states that are arbitrary superpositions of $\Psi^+_\hbar$ and $\Psi^-_\hbar$, i.e.,
\begin{align}
    \Psi_\hbar(t=0) = \alpha \Psi^+_\hbar + \beta \Psi^-_\hbar. \label{eqn: definition of approximately localized states}
\end{align}
This includes the ground state for $\alpha = \beta = \tfrac{1}{\sqrt{2}}$, which Landsman and Reuvers consider, but also includes other states, as well. The $\hbar \to 0$ limits of the two approximately localized states $\Psi^\pm_\hbar$, result in classical states completely localized at $x=\pm a$. So, we will interpret each state as approximating one which, upon measurement, returns $x = \pm a$ with certainty (at least relative to some error bounds).  With this in mind, the Born rule tells us that the superposition state $\Psi_\hbar(t=0)$ above should result in a measurement outcome of $x=+a$ with probability $|\alpha|^2$ and a measurement outcome of $x=-a$ with probability $|\beta|^2$. 

We will show that this is exactly the case when we allow the ``flea" perturbation $\delta$ to be a random variable and the probabilities above to capture the results of measurements involving different fleas. That is, given an arbitrary probability distribution over $\delta$ within a certain class, taking the $t \to \infty$ and $\hbar \to 0$ limits of $\Psi_\hbar(t=0)$ under the $\delta$-perturbed double well time evolution, results in a state localized at $x=+a$ for a set of fleas $\delta$ with probability $|\alpha|^2$ and $x=-a$ for a set of fleas $\delta$ with probability $|\beta|^2$. This is the sense in which Landsman and Reuvers' approach to measurement can also capture the Born rule probabilities. We will eventually show that this holds in the double well toy model, but we begin in \S\ref{subsec:twostate} with a two-state approximation to that model, which will illustrate some of the essential features of our analysis. We will consider the full double well in \S \ref{subsec:doublewell}.

\subsection{Two State Model}
\label{subsec:twostate}
Following Landsman and Reuvers, we begin by working with the “truncated” Hilbert space $\mathbb{C}^2$ spanned by the two lowest energy eigenstates of the system.  To do so, we replace the two approximately localized states $\Psi^-_\hbar$ and $\Psi^+_\hbar$ with 
\begin{align}
    \Phi^- = 
    \begin{pmatrix}
        1 \\ 0
    \end{pmatrix},
    \enspace
    \Phi^+ = 
    \begin{pmatrix}
        0 \\ 1
    \end{pmatrix}
\end{align}
in $\mathbb{C}^2$.  The observables for the system span the C*-algebra $M_2(\mathbb{C})$ of $2\times 2$ matrices. 
 We consider the (unperturbed) Hamiltonian 
\begin{align}
    H_{0,\hbar} = 
        \frac{1}{2}
        \begin{pmatrix}
            0 & -\Delta_\hbar \\
            -\Delta_\hbar & 0 \\
        \end{pmatrix}
\end{align}
where $\Delta_\hbar$ is the difference in the energies of $\Psi^{(1)}_\hbar$ and $\Psi^{(0)}_\hbar$, the energy eigenstates of the double well potential, which we assume has the form given by Eq. (\ref{eqn:Delta_hbar}).  The energy eigenstates of the unperturbed two state Hamiltonian $H_{0,\hbar}$ are
\begin{align}
    \Phi^{(0)}_0 = 
    \frac{1}{\sqrt{2}}
    \begin{pmatrix}
        1 \\ 1
    \end{pmatrix},
    \enspace 
    \Phi^{(1)}_0 = 
    \frac{1}{\sqrt{2}}
    \begin{pmatrix}
        1 \\ -1
    \end{pmatrix}
\end{align}
with respective energies $E^{(0)}_0=-\tfrac{1}{2}\Delta_\hbar$ and $E^{(1)}_0=\tfrac{1}{2}\Delta_\hbar$. These energy eigenstates satisfy the following relationship with the localized states 
\begin{align} 
    \Phi^\pm = \frac{\Phi^{(0)}_0 \mp \Phi^{(1)}_0}{\sqrt{2}}, \label{eqn:two state system localized states}
\end{align}
which mimics the relationship between the two lowest energy states of the double well and the approximately localized states in Equation (\ref{eqn:loc}). The subscripts on the states $\Phi^{(0)}_0$ and $\Phi^{(1)}_0$ signify that they are eigenstates corresponding to the unperturbed Hamiltonian $H_{0,\hbar}$ above.

Now consider a new perturbed Hamiltonian in which a positive real number $\delta$---a “flea”---has been added to the left well, i.e. 
\begin{align} \label{eqn:perturbed hamiltonian}
    H_{\delta,\hbar} = 
        \begin{pmatrix}
            \delta & -\tfrac{1}{2} \Delta_\hbar \\
            -\tfrac{1}{2} \Delta_\hbar & 0 \\
        \end{pmatrix}.
\end{align}
For small $\hbar$, the value of $\delta$ will be large compared to $\Delta_\hbar$ so even though we refer to $\delta$ as a ``perturbation", standard perturbation theory techniques are not applicable.  The resulting energy eigenvalues are \citep[p. 391]{LaRe13} 
\begin{align}
    E^{(0)}_{\delta,\hbar} = \frac{1}{2} \Big( \delta - \sqrt{\delta^2 + \Delta_\hbar^2}\Big) && E^{(1)}_{\delta,\hbar} = \frac{1}{2} \Big( \delta + \sqrt{\delta^2 + \Delta_\hbar^2}\Big).
\end{align}
The energy difference is 
\begin{align}
    \Delta_{\delta,\hbar} = E^{(1)}_{\delta,\hbar} - E^{(0)}_{\delta,\hbar} = \sqrt{\delta^2 + \Delta_\hbar^2}
\end{align} and the corresponding energy eigenstates are
\begin{align}
    \Phi^{(0)}_{\delta,\hbar} = \frac{1}{\sqrt{2}} \Big( \delta^2 + \Delta_\hbar^2 + \delta \sqrt{\delta^2 + \Delta_\hbar^2} \Big)^{-1/2}
        \begin{pmatrix}
            \Delta_\hbar \\
            \delta + \sqrt{\delta^2 + \Delta_\hbar^2}
        \end{pmatrix};
\end{align}
\begin{align}
    \Phi^{(1)}_{\delta,\hbar} = \frac{1}{\sqrt{2}} \Big( \delta^2 + \Delta_\hbar^2 - \delta \sqrt{\delta^2 + \Delta_\hbar^2} \Big)^{-1/2}
        \begin{pmatrix}
            \Delta_\hbar \\
            \delta - \sqrt{\delta^2 + \Delta_\hbar^2}
        \end{pmatrix}.
\end{align}
Notice that $\lim_{\delta \to 0} \Phi^{(i)}_{\delta,\hbar} = \Phi^{(i)}_0$ for $i=0,1$ as we should expect.

As $\hbar \to 0$, $\Delta_\hbar$ rapidly approaches 0 according to Equation (\ref{eqn:Delta_hbar}), and it can be shown that given any $\delta>0$, 
\begin{align}
    \lim_{\hbar \to 0} \Phi^{(0)}_{\delta, \hbar} = \Phi^+, \quad \lim_{\hbar \to 0} \Phi^{(1)}_{\delta,\hbar} = \Phi^-. \label{eqn:convergence of approx localized states}
\end{align}
For $\delta < 0$, these limits are swapped. These results are sensible since the former says that adding a positive perturbation in the left well leads to a ground state which is approximately localized in the right well. They also match onto the eigenvectors of the perturbed Hamiltonian $H_{\delta,\hbar}$ but with 0 substituted for $\Delta_\hbar$. So far, this model is identical to the one treated in \citet{LaRe13}, but we will introduce a novel feature by allowing $\delta$ to be a random variable instead of a real-valued parameter.

Now consider an initial state $\Psi(0) = \alpha \Phi^+ + \beta \Phi^-$ for $\alpha,\beta\in\mathbb{C}$, which at $t=0$ undergoes a measurement of the ``position" observable 
\begin{align}
    Q = \begin{pmatrix}
        -1 & 0 \\
        0 & 1 \\
    \end{pmatrix},
\end{align}
i.e. the observable which assigns the outcome $Q=-1$ to the left well state $\Phi^-$ and the outcome $Q=+1$ to the right well state $\Phi^+$. According to conventional quantum mechanics, the Born rule tells us that a measurement of $Q$ will return $+1$ or $-1$ with respective probabilities
\begin{align}
    P(Q=+1) = |\alpha|^2 && P(Q=-1) = |\beta|^2.
\end{align}
We can describe this Born rule result as the probabilistic mixture
\begin{align} \omega_B = |\alpha|^2 \omega_+ + |\beta|^2 \omega_-
\end{align}
where $\omega_+$ and $\omega_-$ are the expectation value assignments
\begin{align}
\omega_+(A) = \inner{\Phi^+}{A\Phi^+} && \omega_-(A) = \inner{\Phi^-}{A\Phi^-}
\end{align}
for all $A\in M_2(\mathbb{C})$.

The following proposition says that we can recover the Born rule state $\omega_B$ from the ``flea model" of measurement if we take a probabilistic approach to the flea perturbation, under which $\delta$ is a random variable that is constant in time. Each value of $\delta$ determines a distinct time evolution of the initial state $\Psi(0)$ to
\begin{align*}
    \Psi_{\delta,\hbar}(t) = e^{-iH_{\delta,\hbar}t/\hbar}\Psi(t=0).
\end{align*}
The uncertain final state at time $t$ can thus be understood by averaging over different values of $\delta$, weighted by a probability distribution governing the perturbations. We encode this information by considering the time evolved state at time $t$ as a mixed state, i.e. a probabilistic mixture, of the algebraic states of the time evolved pure state for each $\delta$-dependent time evolution. In short, we consider a probability distribution $\mu$ over values of $\delta$, and evaluate the mixed state which combines the various $\delta$-dependent time evolutions of the initial state $\Psi(0)$. We will show that in the $t \to \infty$ and $\hbar \to 0$ limits, the resulting state approaches the Born rule state $\omega_B$ above.

\begin{prop} \label{prop:two state approximation}
Let $\mu$ be a probability measure over real nonzero values of $\delta$ that has a density with respect to the Lebesgue measure over values $\delta\in\R$.  Let $\Psi(0)$ be the initial state
\begin{align}
    \Psi(0) = \alpha \Phi^+ + \beta \Phi^-
\end{align}
for $\alpha,\beta\in\mathbb{C}$ with $|\alpha|^2+|\beta|^2=1$. 
 Consider the time-evolved state $\Psi_{\delta,\hbar}(t) = e^{-iH_{\delta,\hbar}t/\hbar}\Psi(0)$ with the corresponding expectation value assignment $\omega_{\Psi_{\delta,\hbar}(t)}(A) = \inner{\Psi_{\delta,\hbar}(t)}{A\Psi_{\delta,\hbar}(t)}$ for all $A\in M_2(\C)$ . Then the mixed state
\begin{align}
    \omega_\mu^t \equiv \int \omega_{\Psi_{\delta,\hbar}(t)} d \mu(\delta)
\end{align}
approaches the Born rule state $\omega_B$ in the weak* topology under the $t \to \infty$ and then $\hbar \to 0$ limits. That is, for all $A \in M_2(\mathbb{C})$
\begin{align}
\label{eq:twostateBorn}
    \lim_{\hbar \to 0} \lim_{t \to \infty} \big| \omega_\mu^t(A) - \omega_B(A) \big| = 0
\end{align}
\end{prop}

\begin{proof}
First, we will treat the case where $\mu$ has support only over positive values of $\delta$. Arbitrary probability measures (with a density) will be treated at the end of the proof. By definition,
\begin{align}
    \lim_{\hbar \to 0} \lim_{t \to \infty} \big| \omega_\mu^t(A) - \omega_B(A) \big|
    &= \lim_{\hbar \to 0} \lim_{t \to \infty} \bigg | \int \omega_{\Psi_{\delta,\hbar}(t)}(A) d \mu(\delta) - \omega_B(A) \bigg | \\
    &= \lim_{\hbar \to 0} \lim_{t \to \infty} \bigg | \int \langle \Psi_{\delta,\hbar}(t), A \Psi_{\delta,\hbar}(t) \rangle  d \mu(\delta) - |\alpha|^2 \omega_+(A) - |\beta|^2 \omega_-(A)  \bigg | \label{eqn:integral of difference}
\end{align}
We now calculate the above limits for a basis of observables in $M_2(\mathbb{C})$.  We consider the basis consisting of: the right well projection operator $\Pi_+ = |\Phi^+ \rangle \langle \Phi^+ |$, the left well projection operator $\Pi_- = |\Phi^- \rangle \langle \Phi^- |$, and the off-diagonal matrices $A_{+-} = | \Phi^+ \rangle \langle \Phi^-|$ and $A_{-+} = | \Phi^- \rangle \langle \Phi^+|$.  We treat the expectation value of each basis element in turn.

First, setting $A$ to the right well projection operator $\Pi_+ = |\Phi^+ \rangle \langle \Phi^+ |$, we get
\begin{align}
    \lim_{\hbar \to 0} \lim_{t \to \infty} \big| &\omega_\mu^t(\Pi_+) - \omega_B(\Pi_+) \big| \\
    &= \lim_{\hbar \to 0} \lim_{t \to \infty} \bigg | \int \langle \Psi_{\delta,\hbar}(t), \Phi^+ \rangle \langle \Phi^+, \Psi_{\delta,\hbar}(t) \rangle  d \mu(\delta) - |\alpha|^2  \bigg | \\
    &= \lim_{\hbar \to 0} \lim_{t \to \infty} \Bigg |  \int \Big| \big \langle \Phi^+, \Phi_{\delta,\hbar}^{(0)} \big \rangle \Big|^2 \cdot \Big | \big \langle  \Phi_{\delta,\hbar}^{(0)} , \Psi(0) \big \rangle \Big|^2 d \mu(\delta) \label{line:schrodinger evolution:start} \\
    & \hspace{.75in} + \int \Big| \big \langle \Phi^+, \Phi_{\delta,\hbar}^{(1)} \big \rangle \Big|^2 \cdot \Big| \big \langle \Phi_{\delta,\hbar}^{(1)}, \Psi(0) \big \rangle \Big |^2 d \mu(\delta) \\
    & \hspace{.75in} + \int \big \langle \Phi_{\delta,\hbar}^{(0)}, \Phi^+ \big \rangle \big \langle \Phi^+, \Phi_{\delta,\hbar}^{(1)} \big \rangle \big \langle \Psi(0), \Phi_{\delta,\hbar}^{(0)} \big \rangle \big \langle \Phi_{\delta,\hbar}^{(1)}, \Psi(0) \big \rangle  e^{- i \Delta_{\delta,\hbar} t/\hbar} d \mu(\delta) \\
    & \hspace{.75in} + \int \big \langle \Phi_{\delta,\hbar}^{(1)}, \Phi^+ \big \rangle \big \langle \Phi^+, \Phi_{\delta,\hbar}^{(0)} \big \rangle \big \langle \Psi(0), \Phi_{\delta,\hbar}^{(1)} \big \rangle \big \langle \Phi_{\delta,\hbar}^{(0)}, \Psi(0) \big \rangle e^{i \Delta_{\delta,\hbar} t/\hbar} d \mu(\delta) - |\alpha|^2 \Bigg| \label{line:schrodinger evolution:end} \\
    &\leq \Bigg | \lim_{\hbar \to 0} \lim_{t \to \infty} \int \Big| \big \langle \Phi^+, \Phi_{\delta,\hbar}^{(0)} \big \rangle \Big|^2 \cdot \Big | \big \langle  \Phi_{\delta,\hbar}^{(0)} , \Psi(0) \big \rangle \Big|^2 d \mu(\delta) \label{line:triangle inequality} \\
    & \hspace{.5in} + \lim_{\hbar \to 0} \lim_{t \to \infty} \int \Big| \big \langle \Phi^+, \Phi_{\delta,\hbar}^{(1)} \big \rangle \Big|^2 \cdot \Big| \big \langle \Phi_{\delta,\hbar}^{(1)}, \Psi(0) \big \rangle \Big |^2 d \mu(\delta) - |\alpha|^2 \Bigg| 
    \\
    & \hspace{.5in} + \Bigg| \lim_{\hbar \to 0} \lim_{t \to \infty} \int \big \langle \Phi_{\delta,\hbar}^{(0)}, \Phi^+ \big \rangle \big \langle \Phi^+, \Phi_{\delta,\hbar}^{(1)} \big \rangle \big \langle \Psi(0), \Phi_{\delta,\hbar}^{(0)} \big \rangle \big \langle \Phi_{\delta,\hbar}^{(1)}, \Psi(0) \big \rangle  e^{- i \Delta_{\delta,\hbar} t/\hbar} d \mu(\delta) \Bigg| \label{line:vanishing integral 1} \\
    & \hspace{.5in} + \Bigg| \lim_{\hbar \to 0} \lim_{t \to \infty} \int  \big \langle \Phi_{\delta,\hbar}^{(1)}, \Phi^+ \big \rangle \big \langle \Phi^+, \Phi_{\delta,\hbar}^{(0)} \big \rangle \big \langle \Psi(0), \Phi_{\delta,\hbar}^{(1)} \big \rangle \big \langle \Phi_{\delta,\hbar}^{(0)}, \Psi(0) \big \rangle e^{i \Delta_{\delta,\hbar} t/\hbar} d \mu(\delta)  \Bigg| \label{line:vanishing integral 2} \\
    &= \Bigg |  \int \lim_{\hbar \to 0} \Big| \big \langle \Phi^+, \Phi_{\delta,\hbar}^{(0)} \big \rangle \Big|^2 \cdot \Big | \big \langle  \Phi_{\delta,\hbar}^{(0)} , \Psi(0) \big \rangle \Big|^2 d \mu(\delta) \label{line:dominated convergence theorem} \\
    & \hspace{.5in} + \int \lim_{\hbar \to 0} \Big| \big \langle \Phi^+, \Phi_{\delta,\hbar}^{(1)} \big \rangle \Big|^2 \cdot \Big| \big \langle \Phi_{\delta,\hbar}^{(1)}, \Psi(0) \big \rangle \Big |^2 d \mu(\delta) - |\alpha|^2 \Bigg| 
    \\
    &= \Bigg | \int \Big| \big \langle \Phi^+, \Phi^+ \big \rangle \Big|^2 \cdot \Big | \big \langle  \Phi^+ , \Psi(0) \big \rangle \Big|^2 d \mu(\delta) + \int \Big| \big \langle \Phi^+, \Phi^- \big \rangle \Big|^2 \cdot \Big| \big \langle \Phi^-, \Psi(0) \big \rangle \Big |^2 d \mu(\delta) - |\alpha|^2 \Bigg| \label{line:convergence of approx localized states}
    \\
    &= \Bigg | |\alpha|^2 \int d \mu(\delta)  - |\alpha|^2 \Bigg | \label{line:evaluate inner products}
    = 0. 
\end{align}
In lines (\ref{line:schrodinger evolution:start})-(\ref{line:schrodinger evolution:end}), we have used the Schr\"odinger evolution of $\Psi(0)$ and the inequality on the next line (\ref{line:triangle inequality}) follows from the triangle inequality. The integrals on lines (\ref{line:vanishing integral 1}) and (\ref{line:vanishing integral 2}) vanish by Thm. \ref{thm:inftime} because for all $\hbar$, $e^{i \Delta_{\delta,\hbar} t/\hbar}$ weakly approaches a uniform distribution over the unit circle as $t \to \infty$. 
Then, the following equality on line (\ref{line:dominated convergence theorem}) follows from the dominated convergence theorem with a constant unit bounding function. Finally, lines (\ref{line:convergence of approx localized states}) and (\ref{line:evaluate inner products}) follow from Equation (\ref{eqn:convergence of approx localized states}) and from evaluating the inner products, respectively. Since $\Pi_- = 1 - \Pi_+$ it is easy to see that Equation (\ref{eq:twostateBorn}) holds for $A = \Pi_-$ as well.

Now consider the off-diagonal matrix $A_{+-} = | \Phi^+ \rangle \langle \Phi^-|$. Using Equation (\ref{eqn:integral of difference}), we have
\begin{align}
    \lim_{\hbar \to 0} \lim_{t \to \infty} \big| &\omega_\mu^t(A_{+-}) - \omega_B(A_{+-}) \big| \\
    &= \lim_{\hbar \to 0} \lim_{t \to \infty} \bigg | \int \langle \Psi_{\delta,\hbar}(t), \Phi^+ \rangle \langle \Phi^-, \Psi_{\delta,\hbar}(t) \rangle  d \mu(\delta) - 0  \bigg | \\
    &= \lim_{\hbar \to 0} \lim_{t \to \infty} \Bigg |  \int \big \langle \Phi_{\delta,\hbar}^{(0)}, \Phi^+ \big \rangle \big \langle \Phi^-, \Phi_{\delta,\hbar}^{(0)} \big \rangle \Big | \big \langle  \Phi_{\delta,\hbar}^{(0)} , \Psi(0) \big \rangle \Big|^2 d \mu(\delta) \label{line:schrodinger evolution2:start} \\
    & \hspace{.75in} + \int \big \langle \Phi_{\delta,\hbar}^{(1)}, \Phi^+ \big \rangle \big \langle \Phi^-, \Phi_{\delta,\hbar}^{(1)} \big \rangle \Big| \big \langle \Phi_{\delta,\hbar}^{(1)}, \Psi(0) \big \rangle \Big |^2 d \mu(\delta) \\
    & \hspace{.75in} + \int \big \langle \Phi_{\delta,\hbar}^{(0)}, \Phi^+ \big \rangle \big \langle \Phi^-, \Phi_{\delta,\hbar}^{(1)} \big \rangle \big \langle \Psi(0), \Phi_{\delta,\hbar}^{(0)} \big \rangle \big \langle \Phi_{\delta,\hbar}^{(1)}, \Psi(0) \big \rangle  e^{- i \Delta_{\delta,\hbar} t/\hbar} d \mu(\delta) \\
    & \hspace{.75in} + \int \big \langle \Phi_{\delta,\hbar}^{(1)}, \Phi^+ \big \rangle \big \langle \Phi^-, \Phi_{\delta,\hbar}^{(0)} \big \rangle \big \langle \Psi(0), \Phi_{\delta,\hbar}^{(1)} \big \rangle \big \langle \Phi_{\delta,\hbar}^{(0)}, \Psi(0) \big \rangle e^{i \Delta_{\delta,\hbar} t/\hbar} d \mu(\delta) \Bigg| \label{line:schrodinger evolution2:end} \\
    &\leq  \bigg | \lim_{\hbar \to 0} \lim_{t \to \infty} \int  \big \langle \Phi_{\delta,\hbar}^{(0)}, \Phi^+ \big \rangle \big \langle \Phi^-, \Phi_{\delta,\hbar}^{(0)} \big \rangle \Big | \big \langle  \Phi_{\delta,\hbar}^{(0)} , \Psi(0) \big \rangle \Big|^2 d \mu(\delta) \bigg| \label{line:triangle inequality2} \\
    & \hspace{.75in} + \bigg| \lim_{\hbar \to 0} \lim_{t \to \infty} \int \big \langle \Phi_{\delta,\hbar}^{(1)}, \Phi^+ \big \rangle \big \langle \Phi^-, \Phi_{\delta,\hbar}^{(1)} \big \rangle \Big| \big \langle \Phi_{\delta,\hbar}^{(1)}, \Psi(0) \big \rangle \Big |^2 d \mu(\delta) \bigg| \\
    & \hspace{.75in} +  \bigg| \lim_{\hbar \to 0} \lim_{t \to \infty} \int \big \langle \Phi_{\delta,\hbar}^{(0)}, \Phi^+ \big \rangle \big \langle \Phi^-, \Phi_{\delta,\hbar}^{(1)} \big \rangle \big \langle \Psi(0), \Phi_{\delta,\hbar}^{(0)} \big \rangle \big \langle \Phi_{\delta,\hbar}^{(1)}, \Psi(0) \big \rangle  e^{- i \Delta_{\delta,\hbar} t/\hbar} d \mu(\delta) \bigg| \label{line:vanishing integral 1.2} \\
    & \hspace{.75in} + \bigg|  \lim_{\hbar \to 0} \lim_{t \to \infty} \int \big \langle \Phi_{\delta,\hbar}^{(1)}, \Phi^+ \big \rangle \big \langle \Phi^-, \Phi_{\delta,\hbar}^{(0)} \big \rangle \big \langle \Psi(0), \Phi_{\delta,\hbar}^{(1)} \big \rangle \big \langle \Phi_{\delta,\hbar}^{(0)}, \Psi(0) \big \rangle e^{i \Delta_{\delta,\hbar} t/\hbar} d \mu(\delta) \bigg| \label{line:vanishing integral 2.2} \\
    &=  \bigg |  \int \lim_{\hbar \to 0} \big \langle \Phi_{\delta,\hbar}^{(0)}, \Phi^+ \big \rangle \big \langle \Phi^-, \Phi_{\delta,\hbar}^{(0)} \big \rangle \Big | \big \langle  \Phi_{\delta,\hbar}^{(0)} , \Psi(0) \big \rangle \Big|^2 d \mu(\delta) \bigg| \label{line:dominated convergence theorem2:start} \\
    & \hspace{.75in} + \bigg| \int \lim_{\hbar \to 0} \big \langle \Phi_{\delta,\hbar}^{(1)}, \Phi^+ \big \rangle \big \langle \Phi^-, \Phi_{\delta,\hbar}^{(1)} \big \rangle \Big| \big \langle \Phi_{\delta,\hbar}^{(1)}, \Psi(0) \big \rangle \Big |^2 d \mu(\delta) \bigg| \label{line:dominated convergence theorem2:end} \\
    &=  \bigg |  \int  \big \langle \Phi^-, \Phi^+ \big \rangle \big \langle \Phi^-, \Phi^+ \big \rangle \Big | \big \langle  \Phi^- , \Psi(0) \big \rangle \Big|^2 d \mu(\delta) \bigg| \label{line:convergence of approx localized states2} \\
    & \hspace{.75in} + \bigg| \int \big \langle \Phi^+, \Phi^+ \big \rangle \big \langle \Phi^-, \Phi^+ \big \rangle \Big| \big \langle \Phi^+, \Psi(0) \big \rangle \Big |^2 d \mu(\delta) \bigg| \\
    &= 0.
\end{align}

The above calculation follows similarly to the previous one. In lines (\ref{line:schrodinger evolution2:start})-(\ref{line:schrodinger evolution2:end}), we have used the Schr\"odinger evolution of $\Psi(0)$ and the triangle inequality to get line (\ref{line:triangle inequality2}). The integrals on lines (\ref{line:vanishing integral 1.2}) and (\ref{line:vanishing integral 2.2}) vanish again by Thm. \ref{thm:inftime}. The dominated convergence theorem gives (\ref{line:dominated convergence theorem2:start}) and (\ref{line:dominated convergence theorem2:end}). Finally, line (\ref{line:convergence of approx localized states2}) follows from Equation (\ref{eqn:convergence of approx localized states}), and the final line follows from evaluating the inner products. 

The limit similarly vanishes for the other off-diagonal matrix $A_{-+}$ according to an analogous argument. Since these four matrices, $\Pi_+, \Pi_-, A_{+-}$, and $A_{-+}$ span $M_2(\mathbb{C})$, the triangle inequality, along with linearity of $\omega_{\mu}^t$ and $\omega_B$, tells us that for every $A \in M_2(\C)$
\begin{align}
    \lim_{\hbar \to 0} \lim_{t \to \infty} |\omega_\mu^t(A) - \omega_B(A)| = 0.
\end{align}
An analogous argument shows that the above result also holds if $\mu$ is a probability measure with support over only negative values of $\delta$. The only significant difference in this argument is that in the analogue of line (\ref{line:convergence of approx localized states}) the left integrand vanishes and the right integrand becomes 
\begin{align}
    \Big| \big \langle \Phi^+, \Phi^+ \big \rangle \Big|^2  \Big| \big \langle \Phi^+, \Psi(0) \big \rangle \Big |^2,
\end{align}
which gives the same result. Likewise, the same change occurs in line (\ref{line:convergence of approx localized states2}). An arbitrary probability measure satisfying the assumptions of the theorem is a convex combination of a probability measure $\mu_+$ over positive values of $\delta$ and a probability measure $\mu_-$ over negative values of $\delta$. Then, from the definition of $\omega_\mu^t$ we can see that it is a convex combination of $\omega_{\mu_+}^t$ and $\omega_{\mu_-}^t$. So, another appeal to the triangle inequality completes the proof.
\end{proof}

Prop. \ref{prop:two state approximation} shows that the two state model of measurement, when combined with a probabilistic treatment of the flea perturbations, can account for the Born rule probabilities. The method of arbitrary functions in this case thus leads to universal limiting probabilistic behavior.  This shows that, at least in this simple case, the Born rule can be understood as an effective result following from the Schr\"odinger equation and fairly mild assumptions about the distribution of perturbations. In the next subsection we show that a similar result holds, under further assumptions, when treating the full double well potential.

\subsection{Double Well Potential}
\label{subsec:doublewell}

We now turn to the full double well potential, i.e. the potential
    $$ V_0(x) = \frac{1}{4} \lambda (x^2-a^2)^2$$
for $a>0$ and $\lambda >0$, which has minima at $x=\pm a$. Recall that we interpret this model as describing the potential of a binary measurement which takes values $x = \pm a$.  The flea perturbations to this potential are interpreted as uncontrolled contributions from the measuring apparatus or environment during the measuring process. In the last subsection, we showed that if we truncate the Hilbert space (and dynamics) to the two lowest energy eigenstates, the Born rule probabilities for measurements of superposition states can be recovered from the Schr\"odinger time evolution under the description of measurement employed by Landsman and Reuvers using the joint $t\to \infty$ and $\hbar\to 0$ limits. In this section, we show that the Born rule probabilities can also be recovered in the full model if some additional assumptions are satisfied.


The perturbed Hamiltonian corresponding to the above double well potential is
\begin{align}
    H_{\delta,\hbar} = -\frac{\hbar^2}{2m}\frac{\partial^2}{\partial x^2} + \frac{1}{4}\lambda(x^2-a^2)^2 + \delta(x).
    \end{align}
where $\delta$ is a function of position, representing the flea perturbation. The specifics of the flea perturbations $\delta$ will be filled in according to assumptions which require that the ground state and first excited state of the perturbed Hamiltonian converge sufficiently rapidly to particular approximately localized states. Further, we use the notation $\Psi^{(n)}_{\delta,\hbar}$ to denote the $n$th energy eigenstate of the Hamiltonian $H_{\delta,\hbar}$ with perturbation $\delta$ and Planck's constant set to $\hbar>0$.

Consider the ground state $\Psi_{0,\hbar}^{(0)}$ and first excited state $\Psi_{0,\hbar}^{(1)}$ of the unperturbed double well potential. As mentioned previously, their symmetric and antisymmetric superpositions are each then approximately localized in one of the two wells. So, we define
\begin{align}
    \Psi^{\pm}_\hbar = \frac{\Psi_{0,\hbar}^{(0)} \pm \Psi_{0,\hbar}^{(1)}}{\sqrt{2}}
\end{align}
and choose the phases of $\Psi_{0,\hbar}^{(0)}$ and $\Psi_{0,\hbar}^{(1)}$ so that $\Psi^+_\hbar$ is approximately localized in the right well and $\Psi^-_\hbar$ is approximately localized in the left well.\footnote{The sign convention here is different than the one in Eq. (\ref{eqn:two state system localized states}).} For small $\hbar$, each wave function $\Psi^\pm_\hbar$ approximates a Gaussian centered on its respective well.

We now state precisely the previously mentioned assumptions about the space of flea perturbations.  Recall that we aim to treat the flea perturbation $\delta$ as a random variable carrying a probability distribution.  To that end, the assumptions we now present will characterize the space of perturbations in a way that makes the treatment of probability measures over the perturbations tractable. This specific restriction on the space of perturbations is motivated by results of \citet[p. 389-390]{LaRe13} using numerical methods and the WKB approximation \citep[cf. ][]{Ha78,Ha80,JoMaSc81,HeSj84,Si85}. 

Let $\tilde D$ be the set of smooth, real-valued, compactly supported functions on $\mathbb R$.  We will define the set of all flea perturbations of our potential as a subset of $\tilde D$ by introducing a series of restrictions to isolate perturbations of interest. First, we restrict to the two subsets 
\begin{align}
    \tilde{D}^{\pm} =\Big\{\delta\in\tilde{D}\ |\ \lim_{\hbar \to 0} \hbar^{-1} || \Psi^{(0)}_{\delta, \hbar} - \Psi^\pm_\hbar || = 0 \, \text{ and } \lim_{\hbar \to 0} \hbar^{-1} || \Psi^{(1)}_{\delta, \hbar} - \Psi^\mp_\hbar || = 0.\Big\} \label{eqn:convergence of eigenstates to approximately localized states}
\end{align}
where $\Psi^{(0)}_{\delta,\hbar}$ and $\Psi^{(1)}_{\delta,\hbar}$ are the ground state and first excited state, respectively, of the double well potential with a perturbation of $\delta$. 


The subsets $\tilde{D}^\pm$ restrict to fleas for which the resulting ground state and first excited state converge sufficiently rapidly to the approximately localized states in the $\hbar \to 0$ limit.  As mentioned before, \citet[p. 389-390]{LaRe13} provide evidence that this localization occurs for the perturbed energy eigenstates through both numerical methods and the WKB approximation. 
 In what follows, we make this structural assumption about the allowed perturbations, motivated by known mathematical results.

The difference between $\tilde{D}^+$ and $\tilde{D}^-$ lies in which of the two approximately localized states each of the perturbed energy eigenstates converges to.  In $\tilde{D}^+$, we have $\Psi^{(0)}_{\delta,\hbar} \to \Psi^+_\hbar$ and $\Psi^{(1)}_{\delta,\hbar} \to \Psi^-_\hbar$, while in $\tilde{D}^-$, we have $\Psi^{(0)}_{\delta,\hbar} \to \Psi^-_\hbar$ and $\Psi^{(1)}_{\delta,\hbar} \to \Psi^+_\hbar$.  We will treat the general situation by considering both perturbations in $\tilde{D}^+$ and $\tilde{D}^-$.

Recall that to apply the method of arbitrary functions, we need to make some ``smoothness" assumptions about the probability measures we consider over the possible flea perturbations.  In particular, we need to understand the space of possible perturbations as carrying a background measure (analogous to the Lebesgue measure) so that we can restrict to probability measures that possess a density.  To that end, we will provide a numerical representation of our space of perturbations in terms of their corresponding sequences of energy eigenvalues.  This will allow us to treat distributions over (differences of) those energy eigenvalues.

Notice that our measure over the space of perturbations gives rise to a stochastic process given by the random variables $(E^{(n)}_{\delta,\hbar})_{n\in \mathbb{N}}$, each denoting the $n$th energy eigenvalue of the double well Hamiltonian $H_{\delta,\hbar}$ for perturbation $\delta$, i.e.,
\begin{align}
H_{\delta,\hbar}\Psi_{\delta,\hbar}^{(n)} = E^{(n)}_{\delta,\hbar}\Psi_{\delta,\hbar}.
\end{align}
The index set of this stochastic process thus corresponds to the sequence of energy levels $n= 0,1,2,...$ of the system.  In what follows, we will assume that all quantities of interest (e.g., inner products of eigenfunctions) can be thought of as functions that depend only on these energy eigenvalues, or in other words that the space of allowed perturbations can be identified with a space of sequences of corresponding energy eigenvalues.  

To state this more precisely, notice that for each perturbation $\delta$ in $\tilde D$ and $\hbar \in (0,\infty)$, there is a corresponding sequence of real numbers $(E^{(n)}_{\delta,\hbar})_{n\in\N}$ in the space $\R^\N$ of countable sequences of real numbers . Let $\epsilon_\hbar$ be the map $\tilde D \to \mathbb R^{\mathbb N}$ which sends a perturbation $\delta$ to that corresponding sequence of energy eigenvalues in $\R^\N$. We now restrict our space of perturbations $\tilde D^+$ further by requiring that $\epsilon_\hbar$ is injective for all positive $\hbar$ less than some value $\hbar_0$. Label the resulting subspace $D^+$. For every $\hbar < \hbar_0$, we can then use $\epsilon_\hbar$ to identify $D^+$ with some subset $\epsilon_\hbar(D^+) \subset \mathbb{R}^{\mathbb N}$. The space $\epsilon_\hbar(D^+)$ thus consists of sequences of energy eigenvalues for the corresponding perturbed Hamiltonians $H_{\delta,\hbar}$.
The consequence of the assumption that the map $\epsilon_\hbar$ is injective is that every probability measure on $D^+$ can be understood as a measure on some subset of $\R^\N$ and every random variable dependent on $\delta$ can be thought of as a random variable on this same subset of $\R^\N$. Likewise, denote by $D^-$ a corresponding subset of $\tilde{D}^-$ for which the map $\epsilon_\hbar$ is injective.  Finally, note that the set $\R^\N$ carries the countable product of Lebesgue measures on $\R$ as a reference measure, denoted by $\nu$.

The restriction from $\tilde{D}$ to $D^\pm$ can be motivated by the two-state approximation. The original space of perturbations is $\tilde D = \{ \begin{psmallmatrix} \delta & 0 \\ 0 & 0 \end{psmallmatrix} : \delta \in \mathbb R \}$. In this setting, the set of $\delta$ which have perturbed ground state preferring the right well is the set of all $\delta < 0$, so $\tilde D^+$ is the corresponding subset of $\tilde D$. Recall that restricting to $\delta<0$ was an important simplification in the proof of Proposition \ref{prop:two state approximation}. Since $\delta = - \sqrt{\Delta_{\delta,\hbar}^2 - \Delta_\hbar^2}$, we see that the map $\delta \mapsto (E^{(0)}_{\delta, \hbar}, E^{(1)}_{\delta, \hbar})$ is already injective without any further restriction, so $D^+ = \tilde D^+$.  Likewise, $D^-$ corresponds to all values $\delta>0$.  In the two-state approximation, we showed that probability measures over positive $\delta$ and probability measures over negative $\delta$ both lead to states which approach the Born rule state in the combined $t \to \infty$ and $\hbar \to 0$ limits. This meant that convex combinations of such measures would also have this property. A similar strategy yields the same result in the more general case of the double well. 

\begin{thm} \label{thm:double well}
Suppose $\mu$ is a measure on $D^+\cup D^-$ such that for every $\hbar>0$, $\epsilon_\hbar^*\mu$ is absolutely continuous with respect to $\nu$ on $\R^\N$.  Given a state $\Psi_\hbar(t=0) = \alpha \Psi^+_\hbar + \beta \Psi^-_\hbar$, the corresponding time evolved Wigner distributions satisfy
    \begin{align}
        \lim_{\hbar \to 0}  \lim_{t \to \infty}  W^{\Psi_{\delta, \hbar}(t)}_\hbar = |\alpha|^2\cdot  \delta_{(+a,0)} + |\beta|^2\cdot \delta_{(-a,0)}
        \end{align}
with the limits in the weak* topology. Specifically, for every $f \in C_c^\infty(\mathbb R^2)$,
    \begin{align}
    \label{eq:deltalimit}
    \lim_{\hbar \to 0}  \lim_{t \to \infty} \int d \mu(\delta)dx dp \,  f(x,p) \Big( W^{\Psi_{\delta, \hbar}(t)}_\hbar(x,p) - \big( |\alpha|^2 \cdot \delta_{(+a,0)} + |\beta|^2 \cdot\delta_{(-a,0)} \big) \Big) = 0
    \end{align}
\end{thm}

\begin{proof}
It suffices to deal first with the states $\Psi^+_\hbar$ and $\Psi^-_\hbar$, whose Wigner functions are known to converge in the weak* topology to $\delta_{(+a,0)}$ and $\delta_{(-a,0)}$, respectively.  We will show that the following difference vanishes in the limit:
\begin{align}
\label{eq:pmlimit}
    \lim_{\hbar \to 0}  \lim_{t \to \infty} \int d \mu(\delta)dx dp \,  f(x,p) \Big( W^{\Psi_{\delta, \hbar}(t)}_\hbar(x,p) - \big( |\alpha|^2 \cdot W^{\Psi^+}_\hbar + |\beta|^2 \cdot W^{\Psi^-}_\hbar \big) \Big) = 0.
\end{align}

The time evolution of an initial state $\Psi_\hbar(0)$ under the Hamiltonian $H_{\delta,\hbar}$ is given by
\begin{align}
    \Psi_{\delta,\hbar}(t) = \sum_{k=0}^\infty  c_{\delta,\hbar}^{(k)} \Psi^{(k)}_{\delta,\hbar} e^{-i E^{(k)}_{\delta,\hbar} t / \hbar }
\end{align}
where $c_{\delta,\hbar}^{(k)} = \langle \Psi_{\delta, \hbar}^{(k)}, \Psi_\hbar(0) \rangle $. So, then the corresponding Wigner function is
\begin{align}
    W^{\Psi_{\delta,\hbar}(t)}_\hbar(x,p) 
    &= \hbar^{-1} \big \langle \Psi_{\delta,\hbar}(t), \Omega_\hbar^W(x,p) \Psi_{\delta,\hbar}(t) \big \rangle \\
    &= \hbar^{-1} \sum_{j,k=0}^\infty  \overline{ c_{\delta,\hbar}^{(j)}} c_{\delta,\hbar}^{(k)} e^{i (E^{(j)}_{\delta,\hbar} - E^{(k)}_{\delta,\hbar}) t / \hbar }  \big \langle  \Psi^{(j)}_{\delta,\hbar} ,  \Omega_\hbar^W(x,p)    \Psi^{(k)}_{\delta,\hbar}   \big \rangle 
\end{align}
where $\langle \cdot, \cdot \rangle$ is the inner product of $L^2(\mathbb R)$ and $\Omega_\hbar^W(x,p)$ is the linear operator defined by Eq. (\ref{eq:Omega}).

It follows that
\begin{align}
    \lim_{t \to \infty} &\int  d\mu(\delta) dx dp f(x,p)   W_\hbar^{\Psi_{\delta,\hbar}(t)}(x,p) \label{eqn:expectation value of a Wigner distribution} \\
    &= \label{eq:firstline}\hbar^{-1} \sum_{j=0}^\infty \sum_{k=0}^\infty \lim_{t \to \infty}  \int  d\mu(\delta) dx dp f(x,p)   \bigg(  \overline{ c_{\delta,\hbar}^{(j)}} c_{\delta,\hbar}^{(k)} e^{i (E^{(j)}_{\delta,\hbar} - E^{(k)}_{\delta,\hbar})t/\hbar}  \big \langle  \Psi^{(j)}_{\delta,\hbar} ,  \Omega_\hbar^W(x,p)  \Psi^{(k)}_{\delta,\hbar} \big \rangle \bigg) \\ 
    &= \label{eq:secondline} \sum_{k=0}^\infty  \int  d\mu(\delta) dx dp f(x,p)   | c_{\delta,\hbar}^{(k)} |^2 W_\hbar^{\Psi^{(k)}_{\delta,\hbar}} (x,p)  \\
    & \hspace{.75in} +\ \hbar^{-1} \sum_{k=0}^\infty \sum_{j \neq k} \lim_{t \to \infty}  \int  d\mu(\delta) dx dp f(x,p)   \bigg(  \overline{ c_{\delta,\hbar}^{(j)}} c_{\delta,\hbar}^{(k)} e^{i (E^{(j)}_{\delta,\hbar} - E^{(k)}_{\delta,\hbar})t/\hbar}  \big \langle  \Psi^{(j)}_{\delta,\hbar} ,  \Omega_\hbar^W(x,p)  \Psi^{(k)}_{\delta,\hbar} \big \rangle \bigg) \\ \label{eq:thirdline}
    &=  \sum_{k=0}^\infty  \int  d\mu(\delta) dx dp f(x,p)   | c_{\delta,\hbar}^{(k)} |^2 W_\hbar^{\Psi^{(k)}_{\delta,\hbar}} (x,p) \\ 
    &= \label{eq:fourthline} \int  d\mu(\delta) dx dp f(x,p)   \Big(  \sum_{k=0}^\infty | c_{\delta,\hbar}^{(k)} |^2 W_\hbar^{\Psi^{(k)}_{\delta,\hbar}} (x,p) \Big) 
\end{align}
In line (\ref{eq:firstline}) we have appealed to the Fubini-Tonelli theorem and the dominated convergence theorem since $||\Omega_\hbar^W(x,p)|| = 2$ gives $2|f(x,p)|$ as an upper bound on the absolute value of the integrand and since this has a finite integral the Fubini-Tonelli theorem says that we can interchange the sums and integral. In line (\ref{eq:thirdline}) we have used the fact that the $j \neq k$ integrals vanish in the limit $t \to \infty$. This follows from Thm. \ref{thm:inftime} since the factor $e^{i (E^{(j)}_{\delta,\hbar} - E^{(k)}_{\delta,\hbar})t/\hbar}$ is uniformly distributed over the unit circle in this limit. This is shown in detail in Appendix \ref{app:calculations}. Line (\ref{eq:fourthline}) follows from the bound $\big| W_\hbar^{\Psi^{(j)}_{\delta,\hbar}} (x,p) \big| \leq 2 \hbar^{-1}$ and another application of the Fubini-Tonelli theorem.

So we have 
\begin{align}
    &\lim_{\hbar \to 0} \lim_{t \to \infty} \int d\mu(\delta) dx dp f(x,p) \Big( W_\hbar^{\Psi_{\delta,\hbar}(t)}(x,p) - \big(|\alpha|^2 W_\hbar^{\Psi_\hbar^+}(x,p) + |\beta|^2 W_\hbar^{\Psi_\hbar^-}(x,p) \big) \Big) \\
    & \hspace{.5in} = \lim_{\hbar \to 0} \int d\mu(\delta) dx dp f(x,p) \Big( \sum_{k=0}^\infty  \big| c_{\delta,\hbar}^{(k)} \big|^2 W^{\Psi^{(k)}_{\delta,\hbar}}_\hbar(x,p)
    - \big(|\alpha|^2 W_\hbar^{\Psi_\hbar^+}(x,p) + |\beta|^2 W_\hbar^{\Psi_\hbar^-}(x,p) \big) \Big) \\
    & \hspace{.5in} = \lim_{\hbar \to 0} \int d\mu(\delta) dx dp f(x,p) \Big( \big| c_{\delta,\hbar}^{(0)} \big|^2 W^{\Psi^{(0)}_{\delta,\hbar}}_\hbar(x,p)
    - |\alpha|^2 W_\hbar^{\Psi_\hbar^+}(x,p) \Big) \label{eq:firstint} \\
    & \hspace{1in} + \lim_{\hbar \to 0} \int d\mu(\delta) dx dp f(x,p) \Big( \big| c_{\delta,\hbar}^{(1)} \big|^2 W^{\Psi^{(1)}_{\delta,\hbar}}_\hbar(x,p)
    -  |\beta|^2 W_\hbar^{\Psi_\hbar^-}(x,p)  \Big) \label{eq:secondint} \\
    & \hspace{1in} + \lim_{\hbar \to 0} \int d\mu(\delta) dx dp f(x,p) \Big( \sum_{k=2}^\infty  \big| c_{\delta,\hbar}^{(k)} \big|^2 W^{\Psi^{(k)}_{\delta,\hbar}}_\hbar(x,p) \Big) \label{eq:thirdint} \\
    & \hspace{.5in} =0
\end{align}

\noindent
Line (\ref{eq:firstint}) vanishes because
    $$\Big| \big| c_{\delta,\hbar}^{(0)} \big|^2 - |\alpha|^2 \Big| \leq 2||\Psi^{(0)}_{\delta,\hbar} - \Psi^+_\hbar ||, $$ 
    $$\Big| W_\hbar^{\Psi_\hbar^+}(x,p) \Big| \leq 4 \hbar^{-1}, $$ 
    $$\Big| W^{\Psi^{(0)}_{\delta,\hbar}}_\hbar(x,p) - W_\hbar^{\Psi_\hbar^+}(x,p) \Big| \leq 4 \hbar^{-1}  ||\Psi^{(0)}_{\delta,\hbar} - \Psi^+_\hbar ||. \text{}$$
Line (\ref{eq:secondint}) vanishes because analogous inequalities hold for $c_{\delta,\hbar}^{(1)}$, $\beta$, $\Psi^{(1)}_{\delta,\hbar}$, and $\Psi^-_\hbar$.
For line (\ref{eq:thirdint}),
\begin{align}
    \Big| \int d\mu(\delta)dxdp f(x,p) W^{\Psi^{(k)}_{\delta,\hbar}}_\hbar(x,p) \Big| \leq \sup_{\hbar\in[0,1]}\norm{\mathcal{Q}_\hbar(f)}_\hbar
\end{align}
since $\mathcal{Q}_\hbar$ is a strict quantization \cite{Landsman1990a}. So, the dominated convergence theorem yields that 
\begin{align}
    \lim_{\hbar \to 0} \int &d\mu(\delta) \sum_{k=2}^\infty \big| c_{\delta,\hbar}^{(k)} \big|^2 \int dx dp f(x,p)  W^{\Psi^{(k)}_{\delta,\hbar}}_\hbar(x,p) \\
    &= \int d\mu(\delta) \sum_{k=2}^\infty \lim_{\hbar \to 0} \Big( \big| c_{\delta,\hbar}^{(k)} \big|^2  \int dx dp f(x,p)  W^{\Psi^{(k)}_{\delta,\hbar}}_\hbar(x,p) \Big) \\
    &= 0
\end{align}
since $\lim_{\hbar \to 0} \big| c_{\delta,\hbar}^{(k)} \big|^2 = 0$ for $k \geq 2$ which follows from the assumption of Eq. (\ref{eqn:convergence of eigenstates to approximately localized states}). This is because the perturbed energy eigenstates for $k \geq 2$ are orthogonal to the perturbed energy eigenstates for $k=0,1$ which approach the approximately localized states that the initial state is a superposition of.

Hence, we have established Eq. (\ref{eq:pmlimit}).  Now, Eq. (\ref{eq:deltalimit}) follows from the triangle inequality and the fact that
\begin{align}
    \lim_{\hbar\to 0} \int dx dp \, f(x,p) \Big(W^{\Psi^\pm}_\hbar - \delta_{(\pm a, 0)}\Big) = 0.
\end{align}
\end{proof}
        
This theorem states that the $t \to \infty$ and $\hbar \to 0$ limits of the time evolved Wigner distribution $W^{\Psi_{\delta, \hbar}(t)}_\hbar$ lead to a probabilistic mixture of distributions totally localized at $x = \pm a$. Moreover, the coefficients $|\alpha|^2$ and $|\beta|^2$ are exactly the ones we would expect from the Born rule since for small enough $\hbar$, the component $\Psi^+$ represents a state arbitrarily well localized at $x=+a$, and likewise the component $\Psi^-$ represents a state arbitrarily well localized at $x=-a$. Thus, according to the Born rule, measurement of the state $\alpha \Psi^+_\hbar + \beta \Psi^-_\hbar$ should result in a value of $x = +a$ with probability $|\alpha|^2$ and a value of $x = -a$ with probability $|\beta|^2$, as we have found. So, Theorem \ref{thm:double well} says that the model Landsman and Reuvers provide for a binary measurement in the perturbed double well can also capture the Born rule probabilities in the $t\to\infty$ and $\hbar\to 0$ limits.

\section{Conclusion and Future Directions}
\label{sec:con}

In this paper, we have shown that one can reproduce the Born rule by applying the method of arbitrary functions to a toy quantum mechanical model of a measurement setup.  Specifically, we showed that the probabilities predicted by the dynamical evolution of the quantum state under the Schr\"{o}dinger equation for the perturbed double well Hamiltonian come arbitrarily close to the Born rule probabilities for large enough times and small enough values of Planck's constant.  This statement is captured by our results in Prop. \ref{prop:two state approximation} and Thm. \ref{thm:double well} about the successive $t\to\infty$ and $\hbar\to 0$ limits.  We now make a number of remarks about the interpretation and significance of our results.


The explanatory schema that we have proposed does not attempt to derive probabilities from non-probabilistic facts.  In other words, we have not attempted to give an explanation of why probabilities arise within physics in the first place.  Rather, as \citet[][p. 118]{My21} emphasizes, the method of arbitrary functions always begins by assuming the existence of certain probabilities.  We assume that there is an initial probability distribution over some of the parameters---namely, governing the flea perturbation---in the dynamics of the physical system.  From this, we can derive the Born rule probabilities for the outcomes of a measurement in the double well toy model as universal probabilities arising in the $t\to\infty $ and $\hbar\to 0$ limits.  So ultimately, we have derived some probabilities from others.  What is interesting is that in this case the initial distribution over the dynamical parameters can be chosen arbitrarily, and as long as it comes from within our assumed class of distributions, it will lead to the Born rule probabilities in those limits.


Of course, for this explanatory schema to generate an actual explanation, one should rightly ask for an interpretation of those initial probability distributions over the dynamical parameters governing the system.  What is the physical significance of the probability distributions over the flea perturbation?  One option is that they signify a kind of objective or physical probability.  In this case, such a physical probability might be associated with the mechanism by which the flea perturbation is produced as an effective potential from the interaction of some other physical system with the measurement apparatus or double well system.  Another option is that the distributions represent the reasonable beliefs of a rational agent modeling the system who does not have precise knowledge of the form of the perturbation.  Analogous interpretive options have been characterized and discussed in the literature on the classical method of arbitrary functions \citep{Ab10,Ro11a,Ho16,My21,dC22}; that same discussion is warranted when applying the method of arbitrary functions to quantum systems, but we leave it to future interpretive work.  The exact status of the explanation generated by the method of arbitrary functions---both the explanandum and the explanans---may differ depending on what the initial probability distributions are taken to represent.


We take the universal limiting behavior as $t\to\infty$ and $\hbar\to 0$ to provide an approximation by showing that the probabilities for measurement outcomes in the toy model of the double well come arbitrarily close to the Born rule probabilities for large enough times and small enough values of Planck's constant.  But one can do better with the method of arbitrary functions.  Proofs of convergence in limits like those analyzed in this paper involve calculating explicit error bounds on how far the probabilities can deviate from the universal behavior \citep[][p. 108]{My21}.  If one desires those bounds, one can work backwards from mathematical results in the method of arbitrary functions \citep[][\S3]{En92} concerning bounds for random variables whose density has bounded total variation.  Convergence in the appropriate limits encodes the fact that such error bounds must exist.


The explanation schema we have offered here is intended to be neutral among different interpretations of quantum mechanics.\footnote{Compare this with the claim by \citet{St22} that the Deutsch-Wallace derivation of the Born rule is interpretation neutral and thus applies beyond Everettian many worlds theories.}  Since our results rely only on the standard quantum dynamics, they may be employed in any no-collapse interpretation that involves unitary Schr\"{o}dinger evolution, including Bohm-type pilot wave theories and Everettian many worlds theories.  Our results might also be understood to bear on spontaneous collapse 
 dynamics \citep{GhRiWe86,BaGh03}, insofar as those dynamics approximately mimic Schr\"{o}dinger evolution for small enough values of $\hbar$.  As such, we would expect that proponents of any of those interpretations would fill in details of the explanation corresponding to their preferred physical understanding of the measurement process.  Here, our only aim was to set out the explanatory schema, so we leave it as an open question whether any of these more specific interpretations can fill in details compatible with their solution to the measurement problem.  We believe it is worth further investigation to see whether our explanatory schema is compatible with those interpretations.  In this paper though, we refrain from taking on any particular solution to the measurement problem.

We also allow that there may be more than one way to provide an explanation of the Born rule probabilities.  Our results form the backbone of one type of explanation.  Typicality arguments \citep{DuGoZa92} and decision-theoretic approaches \citep{Wa11} might provide others.  We take no position here on whether these should be seen as rivals; we leave open the possibility that alternative explanations might peacefully coexist.


Likewise, in this paper we remain agnostic about the view of \citet{LaRe13} that the flea model provides a novel solution to the measurement problem.  Our claim is that the Born rule can be seen to follow from structural features of the Schr\"{o}dinger dynamics in these measurement-type scenarios.  We have borrowed the dynamical model used by Landsman and Reuvers to illustrate these structural features, but our results are independent of how one interprets the flea dynamics to bear on the appearance of determinate measurement outcomes.

What we take from the discussion of the double well by \citet{LaRe13} is that the flea perturbation leads to a quantum dynamics that evolves initial quantum superpositions so that they become arbitrarily close to states that are approximately localized at the minima of the wells for large enough times and small enough values of Planck's constant.  Those results use only the standard unitary evolution of quantum theory, so any interpretation of quantum mechanics employing the unitary dynamics has access to those facts.  It is a further step to the view that these results suffice to explain the appearance of classical measurement outcomes \citep{LaRe13,La17}.  For the significance of our results, one might subscribe to Landman's view, or one might instead attempt to use a different interpretation of quantum mechanics to fill in details concerning how measurement outcomes appear in this model.  We leave open these possible interpretations and only claim to have derived results that follow from the standard quantum dynamics.


While we have worked in the context of the particular toy model of measurement provided by the perturbed double well potential, the source of our explanatory schema relies only on certain structural features of the dynamics that other models might share.  In our explanatory schema, it is the $t\to\infty$ limit that makes interference terms vanish in the unitary evolution of a superposition of energy eigenstates, thus yielding a mixture of those energy eigenstates.  Then it is the $\hbar\to 0$ limit that brings the resulting energy eigenstates to approximate definite outcome states.  After these two steps, we see that the components of the resulting mixture are the outcome states carrying the Born rule probabilities as weights in a convex sum.


In the first step of our explanatory schema governed by the $t\to \infty$ limit, we rely only on quite general features of quantum mechanics to suppress the interference terms.  Specifically, the double well example illustrates how cross-terms in relevant inner products or integrals of Wigner functions will involve periodic functions whose frequency is proportional to an energy gap.  The method of arbitrary functions will apply whenever the frequency of a time evolution is a random variable.  So whenever it is appropriate to treat the energy gaps of the system as random variables---whose probabilistic nature arises from some uncertainty about the exact form of the Hamiltonian---we conjecture that universal limiting behavior will arise.  Thus, it is likely that the method of arbitrary functions has wider application to other quantum systems.


On the other hand, the second step of our explanatory schema relies on special features of the double well model, and we do not know if those features are present in other models.  It is a special feature of the double well model that in the $\hbar\to 0$ limit, the first two perturbed energy eigenstates come arbitrarily close to the localized outcome states for small enough values of Planck's constant.  Our assumptions require that the difference of these vectors converges to zero in norm sufficiently rapidly as $\hbar\to 0$.  There is good numerical and theoretical evidence for convergence \citep{LaRe13}, but we also make strong assumptions about the rate of convergence, which are required for our results.  Our assumptions can be thought of as encoding constraints on the possible forms of the perturbation to the potential, restricting to only those that have this special property as $\hbar\to 0$.  Alternatively, this can be thought of as a restriction on the initial distributions over possible perturbations, limiting to only those distributions whose support lies on the space of perturbations with the relevant convergence properties.  Either way, the assumption that the perturbed energy eigenstates come arbitrarily close to outcome states in the $\hbar\to 0$ limit  deserves further discussion.  In particular, further investigation is necessary to determine if this is an appropriate constraint on perturbations in physical models of measurement, or if this is a reasonable assumption about the initial distributions over the perturbations.  As such, further investigation of the assumptions governing $\hbar\to 0$ convergence may be tied to how one chooses to interpret the initial probability distributions over the possible perturbations.


However one answers these interpretive questions, this paper establishes novel results governing  universal probabilistic limiting behavior in quantum systems.  We have thus shown that the method of arbitrary functions applies beyond classical physics.  We hope this approach, and the results it yields, engender further discussion and better understanding of physical probabilities in quantum systems.


To conclude, we point out some possible directions for future research expanding on the present results.  First, one might consider toy models that allow for more than two outcome states.  One way to model such a setup is with a multiple-well system such as a cosine potential on a closed interval as considered by \citet[][\S4.2.1]{HeWo18}.  One might also wish to study toy models that allow for measurements of continuous quantities, with a corresponding continuous infinity of outcome states.  One way to model such a system is with a higher dimensional quartic potential (See \citep[][p. 416]{La17} and \citep[][p. 79]{vdV23}).  Finally, one might wish to study more realistic models of measurement \citep{SpHa08,Sp09,AlBaNi13}, like the physical dynamics governing the motion of an electron in a Stern-Gerlach device \citep{PoBaCrGo05,WeWe17}.  Do the results of this paper generalize to derivations of probabilities in any of these other models?  At the very least, we believe the method of arbitrary functions deserves further investigation as an approach to generating explanations of quantum probabilities.

\section{Acknowledgements}

The authors thank Jer Steeger, Samuel C. Fletcher, Conor Mayo-Wilson, Bashir Abdel-Fattah, and audiences at the University of Bristol and the Washington eXperimental Mathematics Lab for helpful discussions concerning the material in this paper. This work was supported by the National Science Foundation under Grant No. 2043089.

\begin{appendices}

\section{Calculations} \label{app:calculations}

In this appendix we show why it is that the off-diagonal integrals in the proof of Theorem \ref{thm:double well} vanish. Recall that $\mu$ is a probability measure on $D^+\cup D^-$ such that for every $\hbar>0$, $\epsilon_\hbar^*\mu$ is absolutely continuous with respect to $\nu$ on $\R^\N$, where $\nu$ is the countable product of Lebesgue measures on $\R$. Let $g$ be the corresponding density for $\epsilon^*_\hbar\mu$ on $\R^\N$ so that we may replace integrals over the variable $\delta$ with respect to the measure $d\mu(\delta)$ by substituting for integrals over the variable $\mathbf{E} = (E^{(j)}_{\delta,\hbar})_{j\in\mathbb{N}}$ denoting a sequence of energy eigenvalues with respect to the measure $g(\mathbf{E}) d\nu(\mathbf{E})$ . Now, for $j \neq k$ we get that
\begin{align}
    \lim_{t \to \infty}  \int d \mu(\delta) &  \bigg(  \overline{ c_{\delta,\hbar}^{(j)}} c_{\delta,\hbar}^{(k)} e^{i (E^{(j)}_{\delta,\hbar} - E^{(k)}_{\delta,\hbar})t/\hbar}  \big \langle  \Psi^{(j)}_{\delta,\hbar} ,  \Omega_\hbar^W(x,p)  \Psi^{(k)}_{\delta,\hbar} \big \rangle \bigg) \\
    &= \lim_{t \to \infty}  \int \prod_{l} dE^{(l)} g( \mathbf{E} )  \bigg(  \overline{ c_{\mathbf{E},\hbar}^{(j)}} c_{\mathbf{E},\hbar}^{(k)} e^{i (E^{(j)}_{\mathbf{E},\hbar} - E^{(k)}_{\mathbf{E},\hbar})t/\hbar}  \big \langle  \Psi^{(j)}_{\mathbf{E},\hbar} ,  \Omega_\hbar^W(x,p)  \Psi^{(k)}_{\mathbf{E},\hbar} \big \rangle \bigg) \\
    &= \lim_{t \to \infty} \int \prod_{l \neq j,k} dE^{(l)} \int dE^{(j)} dE^{(k)} g( \mathbf{E} )  \bigg(  \overline{ c_{\mathbf{E},\hbar}^{(j)}} c_{\mathbf{E},\hbar}^{(k)} e^{i (E^{(j)}_{\mathbf{E},\hbar} - E^{(k)}_{\mathbf{E},\hbar})t/\hbar}  \big \langle  \Psi^{(j)}_{\mathbf{E},\hbar} ,  \Omega_\hbar^W(x,p)  \Psi^{(k)}_{\mathbf{E},\hbar} \big \rangle \bigg) \\ 
    &= \frac{1}{2} \lim_{t \to \infty} \int \prod_{l \neq j,k} dE^{(l)} dv \int du \, g(\mathbf{E})  \bigg(  \overline{ c_{\mathbf{E},\hbar}^{(j)}} c_{\mathbf{E},\hbar}^{(k)} e^{i ut/\hbar}  \big \langle  \Psi^{(j)}_{\mathbf{E},\hbar} ,  \Omega_\hbar^W(x,p)  \Psi^{(k)}_{\mathbf{E},\hbar} \big \rangle \bigg) \label{line:substitutions} \\
    &= \frac{1}{2} \lim_{t \to \infty} \int \prod_{l \neq j,k} dE^{(l)} dv \int du \, \alpha^\hbar_{x,p}( \mathbf{E} ) e^{i ut/\hbar} \label{line:relabeled} \\
    &= \frac{1}{2} \int \prod_{l \neq j,k} dE^{(l)} dv \lim_{t \to \infty} \int du \, \alpha^\hbar_{x,p}( \mathbf{E} ) e^{i ut/\hbar} \label{line:dominated convergence theorem in appendix} \\
    &= \frac{1}{2} \int \prod_{l \neq j,k} dE^{(l)} dv \lim_{t \to \infty} E_\alpha[e^{iut/\hbar}] \\
    &= 0.
\end{align}
In the line (\ref{line:substitutions}) we have used the substitutions $u = E^{(j)}_{\mathbf{E},\hbar} - E^{(k)}_{\mathbf{E},\hbar} $ and $v = E^{(j)}_{\mathbf{E},\hbar} + E^{(k)}_{\mathbf{E},\hbar} $. In the line (\ref{line:relabeled}) we have simply relabeled most of the integrand using the function
\begin{align}
\alpha^\hbar_{x,p}(\mathbf{E}) = g(\mathbf{E}) \overline{ c_{\mathbf{E},\hbar}^{(j)}} c_{\mathbf{E},\hbar}^{(k)} \big \langle  \Psi^{(j)}_{\mathbf{E},\hbar} ,  \Omega_\hbar^W(x,p)  \Psi^{(k)}_{\mathbf{E},\hbar} \big \rangle.
\end{align}
Line (\ref{line:dominated convergence theorem in appendix}) follows from the dominated convergence theorem with the bound
\begin{align}
    \Big| \int du \, \alpha^\hbar_{x,p}( \mathbf{E} ) e^{i ut/\hbar} \Big|
    \leq  \int du \, | \alpha^\hbar_{x,p}( \mathbf{E} ) | 
    \leq 2 \int du |g(\mathbf E)|
\end{align}
which is finite since $g$ is a probability density for $\R^\N$. Finally, the last two steps involve recognizing that the integral over $u$ is simply the expectation value (denoted $E_\alpha$) of the random variable $e^{iut/\hbar}$ for a complex density $\alpha^\hbar_{x,p}$ and then using the fact that this expectation value approaches 0 in the $t \to \infty$ limit \citep{En92}.


\end{appendices}

\newpage
\bibliographystyle{apalike}
\bibliography{bibliography.bib}

\begin{thebibliography}{}

\bibitem[Abrams, 2010]{Ab10}
Abrams, M. (2010).
\newblock {Mechanistic probability}.
\newblock {\em Synthese}, 187(2):343--375.

\bibitem[Allahverdyan et~al., 2013]{AlBaNi13}
Allahverdyan, A.~E., Balian, R., and Nieuwenhuizen, T.~M. (2013).
\newblock Understanding quantum measurement from the solution of dynamical
  models.
\newblock {\em Phys. Rep.}, 525(1):1--166.

\bibitem[Barrett, 2020]{Ba20b}
Barrett, J.~A. (2020).
\newblock {\em {The Conceptual Foundations of Quantum Mechanics}}.
\newblock Oxford University Press.

\bibitem[Bassi and Ghirardi, 2003]{BaGh03}
Bassi, A. and Ghirardi, G. (2003).
\newblock {Dynamical reduction models}.
\newblock {\em Phys. Rep.}, 379(5–6):257--426.

\bibitem[de~Canson, 2022]{dC22}
de~Canson, C. (2022).
\newblock {Objectivity and the Method of Arbitrary Functions}.
\newblock {\em Brit. J. Philos. Sci.}, 73(3):663--684.

\bibitem[Diaconis et~al., 2007]{DiHoMo07}
Diaconis, P., Holmes, S., and Montgomery, R. (2007).
\newblock {Dynamical Bias in the Coin Toss}.
\newblock {\em SIAM Rev.}, 49(2):211--235.

\bibitem[D\"{u}rr et~al., 1992]{DuGoZa92}
D\"{u}rr, D., Goldstein, S., and Zanghi, N. (1992).
\newblock {Quantum Equilibrium and the Origin of Absolute Uncertainty}.
\newblock {\em J. Stat. Phys.}, 67(5):843--907.

\bibitem[Engel, 1992]{En92}
Engel, E. M. R.~A. (1992).
\newblock {\em {A Road to Randomness in Physical Systems}}.
\newblock Springer New York.

\bibitem[Ghirardi et~al., 1986]{GhRiWe86}
Ghirardi, G., Rimini, A., and Weber, T. (1986).
\newblock {Unified Dynamics for Microscopic and Macroscopic Systems}.
\newblock {\em Phys. Rev. D}, 34(2):470.

\bibitem[Griffiths, 1995]{Gr95}
Griffiths, D.~J. (1995).
\newblock {\em Introduction to quantum mechanics}.
\newblock Cambridge University Press, Cambridge, second edition edition.

\bibitem[Harrell, 1978]{Ha78}
Harrell, E.~M. (1978).
\newblock On the rate of asymptotic eigenvalue degeneracy.
\newblock {\em Commun. Math. Phys.}, 60(1):73--95.

\bibitem[Harrell, 1980]{Ha80}
Harrell, E.~M. (1980).
\newblock {Double Wells}.
\newblock {\em Commun. Math. Phys.}, 75(3):239--261.

\bibitem[Helffer and Sj\"ostrand, 1984]{HeSj84}
Helffer, B. and Sj\"ostrand, J. (1984).
\newblock {Multiple Wells in the Semi-Classical Limit I}.
\newblock {\em Commun. Partial Differ. Equ.}, 9(4):337--408.

\bibitem[Hoefer, 2016]{Ho16}
Hoefer, C. (2016).
\newblock {Objective chance: not propensity, maybe determinism}.
\newblock {\em Lato Sensu, Revue de la Soci{\'{e}}t{\'{e}} de philosophie des
  sciences}, 3(1):31--42.

\bibitem[Hopf, 1934]{Ho34}
Hopf, E. (1934).
\newblock {On causality, statistics and probability}.
\newblock {\em Journal of Mathematics and Physics}, 17:51--102.

\bibitem[Jona-Lasinio et~al., 1981a]{JoMaSc81}
Jona-Lasinio, G., Martinelli, F., and Scoppola, E. (1981a).
\newblock {New approach to the semiclassical limit of quantum mechanics: I.
  Multiple tunnelings in one dimension}.
\newblock {\em Commun. Math. Phys.}, 80(2):223--254.

\bibitem[Jona-Lasinio et~al., 1981b]{JoMaSc81a}
Jona-Lasinio, G., Martinelli, F., and Scoppola, E. (1981b).
\newblock {The semiclassical limit of quantum mechanics: A qualitative theory
  via stochastic mechanics}.
\newblock {\em Phys. Rep.}, 77(3):313--327.

\bibitem[Kato, 1995]{Ka95}
Kato, T. (1995).
\newblock {\em {Perturbation theory for linear operators}}.
\newblock Classics in mathematics. Springer, Berlin, repr. of the 1980 ed.
  edition.
\newblock Früher als: Grundlehren der mathematischen Wissenschaften ; 132.

\bibitem[Landsman, 1990]{Landsman1990a}
Landsman, N.~P. (1990).
\newblock {Quantization and Superselection Sectors I. Transformation Group
  C*-Algebras}.
\newblock {\em Rev. Math. Phys.}, 2(1):45--72.

\bibitem[Landsman, 1998]{La98b}
Landsman, N.~P. (1998).
\newblock {\em {Mathematical Topics Between Classical and Quantum Mechanics}}.
\newblock Springer, New York.

\bibitem[Landsman, 2013]{La13}
Landsman, N.~P. (2013).
\newblock {Spontaneous symmetry breaking in quantum systems: Emergence or
  reduction?}
\newblock {\em Stud. Hist. Philos. Mod. Phys.}, 44:379--394.

\bibitem[Landsman, 2017]{La17}
Landsman, N.~P. (2017).
\newblock {\em {Foundations of Quantum Theory: From Classical Concepts to
  Operator Algebras}}.
\newblock Springer.

\bibitem[Landsman and Reuvers, 2013]{LaRe13}
Landsman, N.~P. and Reuvers, R. (2013).
\newblock {A Flea on Schr\"{o}dinger's Cat}.
\newblock {\em Found. Phys.}, 43:373--407.

\bibitem[Lewis, 2016]{Le16}
Lewis, P.~J. (2016).
\newblock {\em {Quantum Ontology: A Guide to the Metaphysics of Quantum
  Mechanics}}.
\newblock Oxford University Press.

\bibitem[Maudlin, 1995]{Ma95}
Maudlin, T. (1995).
\newblock {Three measurement problems}.
\newblock {\em Topoi}, 14(1):7--15.

\bibitem[Maudlin, 2019]{Ma19}
Maudlin, T. (2019).
\newblock {\em {Philosophy of Physics: Quantum Theory}}.
\newblock Princeton University Press.

\bibitem[Myrvold, 2021]{My21}
Myrvold, W.~C. (2021).
\newblock {\em {Beyond Chance and Credence}}.
\newblock Oxford University Press.

\bibitem[Poincar\'{e}, 1896]{Po96}
Poincar\'{e}, H. (1896).
\newblock {\em {Calcul des Probabilit\'{e}s}}.
\newblock Gauthier-Villars, Paris.

\bibitem[Potel et~al., 2005]{PoBaCrGo05}
Potel, G., Barranco, F., Cruz-Barrios, S., and Gómez-Camacho, J. (2005).
\newblock {Quantum mechanical description of Stern-Gerlach experiments}.
\newblock {\em Phys. Rev. A}, 71(5):052106.

\bibitem[Rosenthal, 2011]{Ro11a}
Rosenthal, J. (2011).
\newblock {Probabilities as Ratios of Ranges in Initial-State Spaces}.
\newblock {\em Journal of Logic, Language and Information}, 21(2):217--236.

\bibitem[Ruetsche, 2023]{Ru23}
Ruetsche, L. (2023).
\newblock {UnBorn: Probability in Bohmian Mechanics}.
\newblock {\em Philosophy of Physics}, 1(1):1--23.

\bibitem[Sebens and Carroll, 2018]{SeCa18}
Sebens, C.~T. and Carroll, S.~M. (2018).
\newblock {Self-locating Uncertainty and the Origin of Probability in
  Everettian Quantum Mechanics}.
\newblock {\em Brit. J. Philos. Sci.}, 69(1):25--74.

\bibitem[Simon, 1985]{Si85}
Simon, B. (1985).
\newblock {Semiclassical analysis of low lying eigenvalues. IV. The flea on the
  elephant}.
\newblock {\em J. Funct. Anal.}, 63(1):123--136.

\bibitem[Spehner, 2009]{Sp09}
Spehner, D. (2009).
\newblock {Models for Quantum Measurement}.
\newblock
  \url{https://www-fourier.ujf-grenoble.fr/~spehner/notes_de_cours_q_meas6.pdf}.

\bibitem[Spehner and Haake, 2008]{SpHa08}
Spehner, D. and Haake, F. (2008).
\newblock Quantum measurements without macroscopic superpositions.
\newblock {\em Phys. Rev. A}, 77(5):052114.

\bibitem[Steeger, 2022]{St22}
Steeger, J. (2022).
\newblock {One world is (probably) just as good as many}.
\newblock {\em Synthese}, 200(2).

\bibitem[Strevens, 2011]{St11}
Strevens, M. (2011).
\newblock {Probability Out of Determinism}.
\newblock In Beisbart, C. and Hartmann, S., editors, {\em Probabilities in
  Physics}, pages 339--364. Oxford University Press.

\bibitem[Valentini, 1991]{Va91}
Valentini, A. (1991).
\newblock {Signal-locality, uncertainty, and the subquantum H-theorem. {II}}.
\newblock {\em Phys. Lett. A}, 158(1-2):1--8.

\bibitem[van~de Ven, 2023]{vdV23}
van~de Ven, C. J.~F. (2023).
\newblock Emergent phenomena in nature: A paradox with theory?
\newblock {\em Found. Phys.}, 53(5).

\bibitem[van Heugten and Wolters, 2018]{HeWo18}
van Heugten, J. and Wolters, S. (2018).
\newblock {Obituary for a Flea}.
\newblock In {\em Springer Proc. Math. Stat.}, pages 331--360. Springer
  Singapore.

\bibitem[von Plato, 1983]{vP83}
von Plato, J. (1983).
\newblock {The Method of Arbitrary Functions}.
\newblock {\em Brit. J. Philos. Sci.}, 34(1):37--47.

\bibitem[Wallace, 2011]{Wa11}
Wallace, D. (2011).
\newblock Taking particle physics seriously: A critique of the algebraic
  approach to quantum field theory.
\newblock {\em Stud. Hist. Philos. Mod. Phys.}, 42(2):116--125.

\bibitem[Wallace, 2012]{Wa12}
Wallace, D. (2012).
\newblock {\em {The Emergent Multiverse: Quantum Theory According to the
  Everett Interpretation}}.
\newblock Oxford University Press.

\bibitem[Wennerström and Westlund, 2017]{WeWe17}
Wennerström, H. and Westlund, P.-O. (2017).
\newblock {A Quantum Description of the Stern{\textendash}Gerlach Experiment}.
\newblock {\em Entropy}, 19(5):186.

\end{thebibliography}

\end{document}